%% file: main.tex
\documentclass[runningheads, a4paper]{llncs}

\usepackage{amsmath,amssymb}
\usepackage{mathtools} % \coloneqq symbol
\usepackage{tikz}
\usepackage{ltl}
\usepackage[linesnumbered,lined,commentsnumbered,ruled,noend]{algorithm2e}
\usepackage{color}
\usepackage{colortbl}
\usetikzlibrary{decorations.pathreplacing}
\usetikzlibrary{arrows,petri,backgrounds, positioning, shapes, patterns,calc,automata,intersections,calc}
\usepackage{hyperref}
\usepackage{multirow}
\usepackage{scalerel}
\usepackage{cleveref}
\usepackage{array}
\usepackage{comment}
\usepackage{listings}
\usepackage{marvosym}
\usepackage{stmaryrd}
\usepackage{fdsymbol}
\usepackage{cite}
\usepackage{enumerate}
\usepackage{wrapfig}

\allowdisplaybreaks
\title{Automata-Based Software Model Checking of Hyperproperties\thanks{This work was partially supported by the European Research Council (ERC) Grant HYPER (No. 101055412).}}
\author{Bernd Finkbeiner\inst{1} \and Hadar Frenkel\inst{1} \and Jana Hofmann\inst{2}\thanks{Research carried out while at CISPA Helmholtz Center for Information Security.} \and Janine Lohse\inst{3}}
\institute{CISPA Helmholtz Center for Information Security, Saarbrücken, Germany\\ \and
Azure Research, Microsoft, Cambridge, United Kingdom\\ \and
Saarland University, Saarbrücken, Germany}

\begin{document}

\input{Macros.tex}
\maketitle

\begin{abstract}
    \input{abstract_proposal}
\end{abstract}
\input{introduction.tex} 
\input{related_work.tex}

\input{preliminaries}
\input{tsl_formal.tex}
\input{tsl_model_checking}
\input{infinite_forall_exists}
\input{discussion}

\bibliographystyle{plain}
\bibliography{thesis}

\newpage
\appendix 
\section{Detailed Correctness Proofs} \label{app:correctnessproof}
\input{correctnessproof.tex}
\end{document}

%% file: Macros.tex
%% Macros
\newcommand{\janine}[1]{{\footnotesize{\color{teal} [Janine: #1]}}}
\newcommand{\hadar}[1]{{\footnotesize{\color{purple} [Hadar: #1]}}}

\newcommand{\aut}[1]{\ensuremath{\mathcal{#1}}}

\newcommand{\M}[1]{\ensuremath{\textit{#1}}} % \mathup
\newcommand{\imp}{\mathbin{\rightarrow}~}
\newcommand{\Imp}{\mathbin{\Rightarrow}~}
\renewcommand{\iff}{\mathbin{\leftrightarrow}}
\newcommand{\defeq}{\mathrel{\mathop:=}}

\newcommand{\dd}[2]{\ensuremath{#1_1,\dots,#1_{#2}}}  % x_1,...,x_n
\newcommand{\ddd}[2]{\ensuremath{#1_1\dots\,#1_{#2}}}  % x_1...x_n
\newcommand{\col}{\colon}
\newcommand{\set}[1]{\ensuremath{\{#1\}}}
\newcommand{\mset}[2]{\set{\,#1\mid#2\,}}
\newcommand{\eset}{\ensuremath{\emptyset}}
\newcommand{\incl}{\ensuremath{\subseteq}}
\newcommand{\bnfor}{~|~} %|
\newcommand{\qedsq}{\hfill\ensuremath{_\blacksquare}}

\newcommand{\blankpage}{\newpage
\thispagestyle{empty}
\mbox{}
\newpage}

\newcommand{\Until}{\mathcal{U}}
\newcommand{\TODO}[1]{\colorbox{yellow}{TODO:} \textit{#1}}

%General

\newcommand{\paths}{Paths}
\newcommand{\Lab}{L}
\newcommand{\Path}{\sigma}
\newcommand{\Traces}{\mathcal{L}}
\newcommand{\lmodels}{\models_{LTL}}
\newcommand{\UPredSeq}{Seq}

%TSL

\newcommand{\Upd}[2]{\llbracket #1 \leftarrowtail #2 \rrbracket}
\newcommand{\Inputs}{\mathbb{I}}
\newcommand{\Cells}{\mathbb{C}}
\newcommand{\PredTerm}{\tau_P}
\newcommand{\PredTerms}{\mathcal{T}_P}
\newcommand{\FuncTerm}[1]{\tau_{F #1}}
\newcommand{\FuncTerms}{\mathcal{T}_{\mathcal{F}}}
\newcommand{\UpdTerms}{\mathcal{T}_U}
\newcommand{\Values}{\mathbb{V}}
\newcommand{\FEval}{\varepsilon}
\newcommand{\Eval}{\eta}
\newcommand{\Functions}{\mathcal{F}}
\newcommand{\Assign}{a}
\newcommand{\Assigns}{\mathsf{A}}
\newcommand{\Comput}{\zeta}
\newcommand{\FuncSymbs}{\mathbb{F}}

%HyperTSL

\newcommand{\TraceVs}{\Pi}
\newcommand{\HPredTerm}{\hat{\tau_P}}
\newcommand{\HPredTerms}{\hat{\mathcal{T}_P}}
\newcommand{\HFuncTerm}[1]{\hat{\tau_{F #1}}}
\newcommand{\HFuncTerms}{\hat{\mathcal{T}_F}}
\newcommand{\HAssign}{\hat{a}}
\newcommand{\HAssigns}{\hat{\mathsf{A}}}
\newcommand{\HUpdTerms}{\hat{\mathcal{T}_U}}
\newcommand{\HComput}{\hat{\Comput}}
\newcommand{\ExtComput}[3]{#1 [#2, #3]}
\newcommand{\Computs}{Z}

%Finite Case
\newcommand{\aIn}[2]{in_{#1 = #2}}
\newcommand{\aCell}[2]{cell_{#1 = #2}}
\newcommand{\toLTL}[1]{t_{#1}}
\newcommand{\noUpd}{elimUpd}
\newcommand{\prev}[1]{\overleftarrow{#1}}
\newcommand{\relab}{relabel}
\newcommand{\plugin}[2]{[#1/#2]}
\newcommand{\predTrace}{P}

% TSL Software Model Checking
\newcommand{\feas}[1]{feasible(#1)}
\newcommand{\Stmt}{\mathit{Stmt}}
\newcommand{\BStmt}{\mathit{Stmt}_0}
\newcommand{\PTrace}{\sigma}
\newcommand{\matches}[2]{#1 \smalltriangleleft #2}
\newcommand{\matchest}[3]{#1 \smalltriangleleft_{#3} #2}
\newcommand{\BuechiProd}{\otimes}
\newcommand{\flatten}[1]{\mathit{flatten}(#1)}
\newcommand{\combine}[2]{\mathit{combine}(#1, #2)}
\newcommand{\edgel}{l}
\newcommand{\PredSet}{\rho}
\newcommand{\USet}{\upsilon}

% HyperTSL Software Model Checking

\newcommand{\rememb}{remember}

\renewcommand\qed {{% set up
\parfillskip=0pt % so \par doesnt push \square to left
\widowpenalty=10000 % so we dont break the page before \square
\displaywidowpenalty=10000 % ditto
\finalhyphendemerits=0 % TeXbook exercise 14.32
%
% horizontal
\leavevmode % \nobreak means lines not pages
\unskip % remove previous space or glue
\nobreak % don't break lines
\hfil % ragged right if we spill over
\penalty50 % discouragement to do so
\hskip.2em % ensure some space
\null % anchor following \hfill
\hfill % push \square to right
$\square$% % the end-of-proof mark
%
% vertical
\par}} % build paragraph

%% file: abstract_proposal.tex
%Many important security properties like (generalized) noninterference are \emph{hyperproperties}, that is, properties relating multiple system executions.

%We extend an automata-based LTL software model checking algorithm for hyperproperties with quantifier alternations, developing a sound but incomplete algorithm for finding counterexamples disproving $\forall^*\exists^*$ hyperproperties or witnesses proving $\exists^* \forall^*$ hyperproperties. To the best of our knowledge, this is the first software model checking algorithm that can handle hyperproperties with quantifier alternations without needing a provided finite-state abstraction. As our specification logic, we introduce hyper temporal stream logic modulo theories (HyperTSL(T)) which extends hyper linear temporal logic (HyperLTL) with predicates over inputs and memory cells and with update terms.

We develop model checking algorithms 
for Temporal Stream Logic (TSL) and Hyper Temporal Stream Logic (HyperTSL) modulo theories.
TSL extends Linear Temporal Logic (LTL) with memory cells, functions and predicates, making it a convenient and expressive logic to reason over software and other systems with infinite data domains.
%TSL extends Linear Temporal Logic (LTL) with predicates over inputs and memory cells, and with update terms that specify how the value of a memory cell should change.
%Similar to the extension of LTL to HyperLTL, 
HyperTSL further extends TSL to the specification of hyperproperties -- properties that relate multiple system executions.
% Unlike HyperLTL, 
As such, HyperTSL can express information flow policies like noninterference in software systems.
% , and the update terms make it even more suitable for software verification. 
We augment HyperTSL with theories, resulting in HyperTSL(T), and build on methods from LTL software verification to obtain model checking algorithms for TSL and HyperTSL(T).
This results in a sound but necessarily incomplete algorithm for specifications contained in the $\forall^*\exists^*$ fragment of HyperTSL(T).
Our approach constitutes the first software model checking algorithm for temporal hyperproperties with quantifier alternations that does not rely on a finite-state abstraction.

%Both TSL and HyperTSL were originally defined for synthesis, and this is, to the best of our knowledge, the first attempt at (Hyper)TSL model checking. We study (Hyper)TSL model checking of infinite-state systems, also called software model checking. We first extend HyperTSL with theories, resulting in HyperTSL(T). We then adapt methods from LTL software model checking to the setting of TSL(T) -- TSL modulo theories, further extending them to alternation-free HyperTSL(T). We then provide an algorithm for finding counterexamples for $\forall^* \exists^*$-HyperTSL(T) formulas (or, dually, witnesses of $\exists^* \forall^*$ formulas). 

%% file: introduction.tex
\section{Introduction}\label{sec:intro}
Hyperproperties~\cite{DBLP:journals/jcs/ClarksonS10} generalize trace properties~\cite{as85} to system properties, i.e., properties that reason about a system in its entirety and not just about individual execution traces.
% and not just based on individual executions. 
% Hyperproperties have been shown to be a powerful tool for expressing and reasoning about 
Hyperproperties comprise many important properties that are not expressible as trace properties, e.g.,
% They constitute a unifying reasoning framework for
information flow policies~\cite{DBLP:journals/jcs/ClarksonS10}, sensitivity and robustness of cyber-physical systems, and linearizability in distributed computing~\cite{bss18}.
For software systems, typical hyperproperties are program refinement or fairness conditions such as symmetry.

For the specification of hyperproperties, Linear Temporal Logic \cite{LTL} (LTL) has been extended with trace quantification, resulting in Hyper Linear Temporal Logic \cite{HyperLTL} (HyperLTL). 
There exist several model checking algorithms for HyperLTL~\cite{HyperLTL, HyperLTLModelChecking, VerifyingHyperliveness}, but they are designed for finite-state systems and are therefore not directly applicable to software. Existing algorithms for software verification of temporal hyperproperties (e.g.,~\cite{RelationalCorrectnessProofs, RelationalCorrectnessProofs2}) are, with the exception of \cite{RavensPaper}, limited to universal hyperproperties, i.e., properties without quantifier alternation.

In this paper, we develop algorithms for model checking software systems against $\forall^*\exists^*$ hyperproperties.
Our approach is complementary to the recently proposed approach of \cite{RavensPaper}. They require to be given a finite-state abstraction of the system, based on which they can both prove and disprove $\forall^*\exists^*$ hyperproperties. We do not require abstractions and instead provide sound but necessarily incomplete approximations to detect counterexamples of the specification.
% To do so, we employ a feasibility analysis. 

% While many hyperproperties are expressible using universal quantifiers only, some especially important ones require a combination of both. 
The class of $\forall^*\exists^*$ hyperproperties contains many important hyperproperties like program refinement or \emph{generalized noninterference}~\cite{noninference}.
% One example is \textit{Generalized noninterference} \cite{noninference}, 
Generalized noninterference states that it is impossible to infer the value of a high-security input by observing the low-security outputs.
Unlike \textit{noninterference}, it does not require the system to be deterministic.
Generalized noninterference can be expressed as $\varphi_{gni} = \forall \pi\exists \pi'.\LTLglobally (i_{\pi'} = \lambda \wedge c_{\pi} = c_{\pi'})$.
The formula states that replacing the value of the high-security input $i$ with some dummy value $\lambda$ does not change the observable output $c$.

The above formula can only be expressed in HyperLTL if $i$ and $c$ range over a finite domain.
This is a real limitation in the context of software model checking, where variables usually range over infinite domains like integers or strings.
To overcome this limitation, our specifications build on Hyper Temporal Stream Logic (HyperTSL)~\cite{HyperTSL}.
HyperTSL replaces HyperLTL's atomic propositions with memory cells together with predicates and update terms over these cells. Update terms use functions to describe how the value of a cell changes from the previous to the current step.
This makes the logic especially suited for specifying software properties.

HyperTSL was originally designed for the synthesis of software systems, which is why all predicates and functions are uninterpreted.
In the context of model checking, we have a concrete system at hand, so we should interpret functions and predicates according to that system. 
We therefore introduce HyperTSL(T) -- HyperTSL with interpreted theories -- as basis for our algorithms.

\paragraph{Overview}\input{overview.tex}
    
\paragraph{Contributions.}
We present an automata-based algorithm for software model checking of $\forall^*\exists^*$-hyperproperties. We summarize our contributions as follows. 
\begin{itemize}
    \item We extend HyperTSL with theories, a version of HyperTSL that is suitable for model checking.
    \item  We adapt the approach of~\cite{FairnessModTheory} to TSL(T) and alternation-free HyperTSL(T), and thereby suggest the first model checking algorithm for both TSL(T) and HyperTSL(T). 
    \item We further extend the algorithm for disproving $\forall^*\exists^*$ hyperproperties and proving $\exists^*\forall^*$ hyperproperties using a feasibility analysis. 
\end{itemize}

%% file: overview.tex
%\subsubsection{Overview}
%We now elaborate more on our algorithm and the adaptation of the algorithm of~\cite{FairnessModTheory} to our setting. 
Following \cite{SoftwareModelCheckingAutomata}, we represent our system as a symbolic automaton labeled with program statements. 
Not every trace of such an automaton is also a valid program execution: for example, a trace $\mathit{assert}(n = 0)~; n --;~ (\mathit{assert}(n = 0))^\omega$ \footnote{The superscript $\omega$ denotes an infinite repetition of the program statement.} cannot be a program execution, as the second assertion will always fail. Such a trace is called \textit{infeasible}. In contrast, in a feasible trace, all assertions can, in theory, succeed.
As a first step, we tackle TSL model checking (Sec.~\ref{sec:tslMC}) by constructing a program automaton whose feasible accepted traces correspond to program executions that violate the TSL specification. To do so, we adapt the algorithm of~\cite{FairnessModTheory}, which constructs such an automaton for LTL, combining the given program automaton and an automaton for the 
negated specification.
% To finally test whether a feasible trace exists, we apply the partial algorithm of \cite{FairnessModTheory}.

 We then extend this algorithm for HyperTSL(T) formulas without quantifier alternation (Sec.~\ref{sec:altfree}) by applying \textit{self-composition}, a technique commonly used for the verification of hyperproperties \cite{SelfComposition1, SelfComposition2, SelfComposition3}. %The program automaton is composed with itself $n$ times, where $n$ is the number of quantifiers, and a trace of the resulted automaton corresponds to $n$ traces of the original system.

Next, in Sec.~\ref{sec:hyperMC}, we further extend this algorithm to finding counterexamples for $\forall^*\exists^*$-HyperTSL(T) specifications (and, dually, witnesses for $\exists^*\forall^*$ formulas). 
We construct an automaton that {over-approximates} the combinations of program executions that satisfy the existential part of the formula. If some program execution is not included in the over-approximation, this execution is a counterexample proving that the program violates the specification.

More concretely, for a HyperTSL(T) formula $\forall^m\exists^n \psi $, we construct the product of the automaton for $\psi$ and the $n$-fold self-composition of the program automaton. Every feasible trace of this product corresponds to a choice of executions for the variables $\pi_1, \dots, \pi_n$ such that $\psi$ is satisfied. Next, we remove (some) spurious witnesses by removing infeasible traces. We consider two types of infeasibility: \textit{$k$-infeasibility}, that is, a local inconsistency in a trace appearing within $k$ consecutive timesteps; and infeasibility that is not local, and is the result of some \textit{infeasible accepting cycles} in the automaton.
In the next step, we project the automaton to the universally quantified traces, obtaining an over-approximation of the trace combinations satisfying the existential part of the formula.
Finally, all that remains to check is whether the over-approximation includes all combinations of feasible traces.

Lastly, in Sec.~\ref{sec:HyperExamples}, we demonstrate our algorithm for two examples, including generalized noninterference.

%% file: related_work.tex
\subsubsection{Related Work}
Temporal stream logic extends linear temporal logic \cite{LTL} and was originally designed for synthesis \cite{TSL}. For synthesis, the logic has been successfully applied to synthesize the FPGA game `Syntroids' \cite{Syntroids}, and to synthesize smart contracts \cite{TSLSmartContracts}. To advance smart contract synthesis, TSL has been extended to HyperTSL in \cite{HyperTSL}. %, thereby adding the possibility to relate multiple program executions. 
The above works use a version TSL that leaves functions and predicates uninterpreted. While this choice is very well suited for the purpose of synthesis, for model checking it makes more sense to use the interpretation of the program at hand. % having interpreted predicates is more desirable for other purposes like model checking. 
TSL was extended with theories in~\cite{TSLDec}, which also analyzed the satisfiability problem of the logic. %Our definition differs slightly from theirs: Finkbeiner et al. define the satisfaction of an update term by syntactic comparison of the current program statement and the update term. Thus, for example, the program statement $c := c + 1$ would satisfy the update term $\Upd{c}{c+1}$ while the statement $c := c + 2 - 1$ would not. In model checking, we usually do not want to reason about the syntax of a program, thus in this thesis, we interpret update terms based on the current and the previous value assignment, as done in \cite{TSL(T)}.
Neither TSL nor HyperTSL model checking has been studied so far (with or without interpreted theories).
%As TSL can encode programs, a satisfiability checker can also be used for model checking. However, an algorithm explicitly designed for model checking is likely to be more efficient. The reason for that is that encoding the program in the formula and then translating the formula to an automaton already leads to an automaton with a size exponential in the size of the program. \\

For LTL, the model checking problem for infinite-state models has been extensively studied, examples are~\cite{L2SIA-WFR, IC3, IC3SoftwareModelChecking, FairnessModTheory, DBLP:journals/jar/FrenkelGS19}.
Our work builds on the automata-based LTL software model checking algorithm from \cite{FairnessModTheory}.
There are also various algorithms for verifying universal hyperproperties on programs, for example, algorithms based on type theory \cite{RelationalCorrectnessProofs, RelationalCorrectnessProofs2}.
Major related work is \cite{RavensPaper}, which (in contrast to our approach) requires on predicate abstractions to model check software against $\forall^*\exists^*$ HyperLTL specifications.
They can also handle asynchronous hyperproperties, which is currently beyond our scope.
Another proposal for the verification of $\forall\exists$ hyperproperties on software is~\cite{FunctionalHyperproperties}. Here, generalized constrained horn clauses are used to verify functional specifications. The approach is not applicable to reactive, non-terminating programs.
Recently, it was also proposed to apply model checkers for TLA (a logic capable of expressing software systems as well as their properties) to verify $\forall^*\exists^*$ hyperproperties~\cite{DBLP:conf/csfw/LamportS21}.

Beyond the scope of software model checking, the verification of hyperproperties has been studied for various system models and classes of hyperproperties.
There exist model checking algorithms for $\omega$-regular properties~\cite{HyperLTLModelChecking,DBLP:conf/vmcai/Finkbeiner21} and asynchronous hyperproperties~\cite{DBLP:conf/cav/BaumeisterCBFS21, DBLP:conf/concur/BozzelliPS22} in finite-state Kripke structures, as well as timed systems~\cite{DBLP:conf/time/HoZ019}, real-valued~\cite{DBLP:conf/memocode/NguyenKJDJ17} and probabilistic hyperproperties~\cite{DBLP:conf/spin/AroraHLLP22,DBLP:conf/atva/DimitrovaFT20,DBLP:conf/fm/DobeABB21} (some of which study combinations of the above).

%% file: preliminaries.tex
\section{Preliminaries}

\paragraph{A Büchi Automaton}
   is a tuple $\aut{A} = (\Sigma, Q, \delta, q_0, F)$ where $\Sigma$ is a finite alphabet; $Q$ is a set of states; $\delta \subseteq Q \times \Sigma \times Q$ is the transition relation; $q_0\in Q$ is the initial state; and $F\subseteq Q$ is the set of accepting states.
    A \emph{run} of the Büchi automaton $\aut{A}$ on a word $\sigma \in \Sigma^\omega$ is an infinite sequence $q_0~q_1~q_2 \dots \in Q^\omega$ of states such that for all $i \in \mathbb{N}, (q_i, \sigma_i, q_{i+1}) \in \delta$. An infinite word $\sigma$ is \emph{accepted} by $\aut{A}$ if there is a run on $\sigma$ with infinitely many $i \in \mathbb{N}$ such that $q_i \in F$. The language of $\aut{A}$, $\mathcal{L}(\aut{A})$, is the set of words accepted by $\aut{A}$.

%% file: tsl_formal.tex
\label{ch:TSL}

%\section{(Hyper-) Temporal Stream Logic + Theories}
\subsection{Temporal Stream Logic Modulo Theories TSL(T)}\label{prelim:tsl}
Temporal Stream Logic (TSL) \cite{TSL} extends Linear Temporal Logic (LTL) \cite{LTL} 
by replacing Boolean atomic propositions with predicates over memory cells and inputs, and with \textit{update terms} that specify how the value of a cell should change. 

We present the formal definition of TSL modulo theories -- TSL(T), based on the definition of~\cite{TSLDec}, which extends the definition~\cite{TSL}. The definition we present is due to~\cite{TSL(T)} and it slightly differs from the definition of~\cite{TSLDec}; The satisfaction of an update term is not defined by syntactic comparison, but relative to the current and previous values of cells and inputs. This definition suites the setting of model checking, where a concrete model is given.

TSL(T) is defined based on a set of \textit{values} $\Values$ with $\mathit{true}, \mathit{false}\in \Values$, a set of inputs $\Inputs$ and a set of memory cells $\Cells$.  Update terms and predicates are interpreted with respect to a given theory. A \emph{theory} is a tuple $(\FuncSymbs, \FEval)$, where $\FuncSymbs$ is a set of function symbols; $\FuncSymbs_n$ is the set of functions of arity $n$; %, each one with an arity $n$; 
and $\FEval: \left( \bigcup_{n\in\mathbb{N}} \FuncSymbs_n \times \Values^n \right) \rightarrow \Values$ is the interpretation function, evaluating a function with arity $n$.
For our purposes, we assume that every theory $(\FuncTerms, \FEval)$ contains at least $\{=, \vee, \neg\}$ with their usual interpretations. 
%the equals-predicate, disjunction and negation. 
% Next, we formally present how function terms, predicate terms and how TSL(T) formulas are constructed using function symbols, cells and inputs

A \emph{function term} $\FuncTerm{}$ is defined by the grammar
    \begin{align*}
       \FuncTerm{} ::= c~|~i~|~f (\FuncTerm{},~\FuncTerm{},~\dots~\FuncTerm{})
   \end{align*}
    where $c \in \Cells, i \in \Inputs, f \in \FuncSymbs$, and the number of elements in $f$ matches its arity. An \emph{assignment} $\Assign : (\Inputs \cup \Cells) \rightarrow \Values$ is a function assigning values to inputs and cells. We denote the set of all assignments by $\Assigns$.
Given a concrete assignment, we can compute the value of a function term.

The \emph{evaluation function} $\Eval: \FuncTerms \times \Assigns \rightarrow \Values$ is defined as
% $\Eval(c, \Assign) =  \Assign(c)$ for $c \in \Cells$;  $\Eval(i, \Assign) =  \Assign(i)$ for $i \in \Cells$; and for $ f \in \mathbb{F}$ we have
    \begin{align*}
       \Eval(c, \Assign) &=  \Assign(c) &&\text{for } c \in \Cells \\
       \Eval(i, \Assign) &= \Assign(i) &&\text{for } i \in \Inputs \\
     \Eval(f~(\FuncTerm{1}, \FuncTerm{2}, \dots ,\FuncTerm{n}), a) &= \FEval(f, (\Eval(\FuncTerm{1}), \Eval(\FuncTerm{2}), \dots , \Eval(\FuncTerm{n}) )) && \text{for } f \in \mathbb{F}      
    \end{align*}

A \emph{predicate term} $\PredTerm$ is a function term only evaluating to \emph{true} or \emph{false}.
    %\vspace{-0.3cm}
    %\begin{align*}
    %\forall \Assign \in \Assigns.~ \Eval(\PredTerm, \Assign) = true \vee \Eval(\PredTerm, \Assign) = %false
    %\end{align*}
    We denote the set of all predicate terms by $\PredTerms$.
    
 For $c \in \Cells$ and $\FuncTerm{} \in \FuncTerms$, $\Upd{c}{\FuncTerm{}}$ is called an \emph{update term}.
 Intuitively, the update term $\Upd{c}{\FuncTerm{}}$ states that $c$ should be updated to the value of $\FuncTerm{}$. If in the previous time step $\FuncTerm{}$ evaluated to $v \in \Values$, then in the current time step $c$ should have value $v$. The set of all update terms is $\UpdTerms$.
TSL formulas are constructed as follows, for 
$c \in \Cells, \PredTerm \in \PredTerms, \FuncTerm{} \in\FuncTerms$. %A TSL(T) formula is defined by the grammar:
$$\varphi ::= \PredTerm~|~\Upd{c}{\FuncTerm{}}~|~\neg \varphi~|~\varphi\wedge\varphi~|~\LTLnext \varphi~|~\varphi\LTLuntil\varphi~$$
The usual operators $\vee, \LTLeventually$ (``eventually"), and $\LTLglobally$ (``globally") can be derived using the equations
$\varphi \vee \psi = \neg (\neg \varphi \wedge \neg \psi),~\LTLeventually \varphi = \mathit{true}\,\LTLuntil\varphi$ and $\LTLglobally \varphi = \neg \LTLeventually \neg \varphi$. 

Assume a fixed initial variable assignment $\Comput_{-1}$ (e.g., setting all values to zero). The satisfaction of a TSL(T) formula with respect to a \textit{computation} $\Comput \in \Assigns^\omega$ and a time point $t$ is defined as follows, where we define $\Comput \models \varphi$ as $0, \Comput \models \varphi$. 
\begin{align*}
    &t, \Comput \models \PredTerm &&\Leftrightarrow \Eval(\PredTerm,\Comput_t) = \mathit{true} \\
    &t, \Comput \models \Upd{c}{\FuncTerm{}} &&\Leftrightarrow \Eval(\FuncTerm{}, \Comput_{t-1}) = \Comput_t(c) \\
    &t, \Comput \models \neg \varphi &&\Leftrightarrow \neg (t, \Comput \models \varphi) \\
    &t, \Comput \models \varphi \wedge \psi &&\Leftrightarrow t, \Comput \models \varphi \text{ and } t, \Comput \models \psi \\
    &t, \Comput \models \LTLnext \varphi &&\Leftrightarrow t + 1, \Comput \models \varphi \\
    &t, \Comput \models \varphi\LTLuntil\psi &&\Leftrightarrow \exists t' \geq t.~ t', \Comput \models \psi \text{ and } \forall  t \leq t'' < t'.~ t'', \Comput \models \varphi
\end{align*}

\section{HyperTSL Modulo Theories} \label{sec:HyperTSL} 
In this section, we introduce HyperTSL(T), HyperTSL with theories, which enables us to interpret predicates and functions depending on the program at hand.
In~\cite{HyperTSL}, two versions of HyperTSL are introduced: HyperTSL and HyperTSL$_{rel}$. The former is a conservative extension of TSL to hyperproperties, meaning that predicates only reason about a single trace.
In HyperTSL$_{rel}$, predicates may relate multiple traces, which opens the door to expressing properties like noninterference in infinite domains.
% which does not allow relating multiple traces within one predicate; and HyperTSL$_{rel}$, which allowes it. Many important security properties, noninterference among them, are only expressible using HyperTSL$_{rel}$. %Nevertheless, the authors of \cite{HyperTSL} focused on the more restrictive version to be able to handle the synthesis problem. 
Here, we build on HyperTSL$_{rel}$, allowing, in addition,  update terms ranging over multiple traces. %(and not only predicate terms). %as this version is more suitable for model checking. 
 Furthermore, we extend the originally uninterpreted functions and predicates with an interpretation over theories. We denote this logic by HyperTSL(T).
 
The syntax of HyperTSL(T) is that of TSL(T), with the addition that cells and inputs are now each assigned to a trace variable that represents a computation. For example, $c_{\pi}$ now refers to the memory cell $c$ in the computation represented by the trace $\pi$. Formally, let $\TraceVs$ be a set of trace variables. We define a \textit{hyper-function term} $\HFuncTerm{} \in \HFuncTerms$ as a function term using $(\Inputs \times \TraceVs)$ as the set of inputs and $(\Cells \times \TraceVs)$ as the set of cells.
\begin{definition}
    A \emph{hyper-function term} $\HFuncTerm{}$ is defined by the  grammar 
    $$
        \HFuncTerm{} ::= c_\pi~|~i_\pi~|~f (\HFuncTerm{},~\HFuncTerm{},~\dots~\HFuncTerm{})
    $$
    where $c_\pi \in \Cells \times \TraceVs, i_\pi \in \Inputs \times \TraceVs, f \in \FuncSymbs$, and the number of the elements in the tuple matches the function arity. We denote by $\HFuncTerms$  the set of all hyper-function terms.
\end{definition}

Analogously, we define \textit{hyper-predicate terms} $\HPredTerm \in \HPredTerms$ as hyper-function terms evaluating to \emph{true} or \emph{false}; \textit{hyper-assignments} $\HAssigns = (\Inputs \cup \Cells) \times \TraceVs \rightarrow \Values$ as functions mapping cells and inputs of each trace to their current values; \textit{hyper-computations} $\HComput \in \HAssigns^{\omega}$ as hyper-assignment sequences. See Fig. \ref{fig:first_system_example} for an~example.

\begin{figure}[t]
\begin{minipage}{0.35\textwidth}
\input{TSL_Kripke_structure.tex}
\label{fig:first_system_example}
\end{minipage}%
\begin{minipage}{0.65\textwidth}
\begin{align*}
&\pi := (c=0)~ (c=1)^\omega,~
\pi' := ((c=0)~ (c=1)~ (c=2))^\omega \\
&\HAssign_1: \{c_\pi \mapsto 0, c_{\pi'} \mapsto 0\}, \HAssign_2 : \{c_\pi \mapsto 1, c_{\pi'} \mapsto 1\} \\
&\HAssign_3: \{c_\pi \mapsto 1, c_{\pi'} \mapsto 2\}, \HAssign_4 : \{c_\pi \mapsto 1, c_{\pi'} \mapsto 0\} \\
\end{align*}
\end{minipage}
\caption{Left: A program automaton. Right: two traces $\pi$ and $\pi'$ of the program automaton. We interpret each trace as a computation. When executing both traces simultaneously, every time point has a corresponding hyper-assignment that assigns values to $c_{\pi}$ and $c_{\pi'}$. Those for the first four time steps are shown on the right. Together, they define the hyper-computation $\HComput := \HAssign_1(\HAssign_2~\HAssign_3~\HAssign_4)^\omega$, matching $\pi$ and $\pi'$.}
\end{figure}
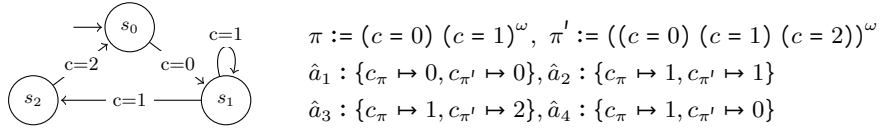

\begin{definition}
     Let  $c_\pi \in \Cells \times \TraceVs, \HPredTerm \in \HPredTerms, \HFuncTerm{} \in \HFuncTerms$. A \emph{HyperTSL(T) formula} is defined by the following grammar:
\begin{align*}
    \varphi &::= \psi~|~\forall \pi.~ \varphi~|~\exists \pi.~\varphi \\
    \psi &::= \HPredTerm~|~\Upd{c_\pi}{\HFuncTerm{}}~|~\neg \psi~| ~\psi\wedge\psi~|~\LTLnext \psi~|~\psi\LTLuntil\psi~
\end{align*}
\end{definition}
To define the semantics of HyperTSL(T), we need the ability to extend a hyper-computation to new trace variables, one for each path quantifier. 
Let $\HComput\in\HAssigns^\omega$ be a hyper-computation, and let
$\pi, \pi' \in \TraceVs, \Comput \in \Assigns^{\omega} $ and $ x \in (\Inputs \cup \Cells)$. We define the extension of $\HComput$ by $\pi$ using the computation $\Comput$ as $ \ExtComput{\HComput}{\pi}{\Comput}~(x_{\pi'})= \HComput(x_{\pi'}) $ for $ \pi' \neq \pi$, and $\ExtComput{\HComput}{\pi}{\Comput}~(x_\pi)= \Comput(x_\pi)$ for $\pi$.

\begin{definition}
    The \emph{satisfaction of a HyperTSL(T)-Formula} w.r.t. a hyper- computation $\HComput \in \HAssigns^\omega$, a set of computations $\Computs$ and a time point $t$ is defined by
    \begin{align*}
        &t, Z, \HComput \models \forall \pi.~\varphi &&\Leftrightarrow \forall \Comput \in \Computs.~t,~Z,~\ExtComput{\HComput}{\pi}{\Comput} \models \varphi \\
        &t, Z, \HComput \models \exists \pi.~\varphi &&\Leftrightarrow \exists \Comput \in \Computs.~t,~Z,~\ExtComput{\HComput}{\pi}{\Comput} \models \varphi \\
    \end{align*}
    The cases that do not involve path quantification are analogous to those of TSL(T) as defined in Sec.~\ref{prelim:tsl}. We define $Z \models \varphi$ as $0, Z, \emptyset^\omega \models \varphi$.
\end{definition}

%% file: TSL_Kripke_structure.tex
% \begin{center}
%     \begin{tikzpicture}[scale=0.2]
%     \tikzstyle{every node}+=[inner sep=0pt]
%     \draw [black] (4,-8) circle (3);
%     \draw (4,-8) node {$\aCell{c}{0}$};
%     \draw [black] (24.7,-3.2) circle (3);
%     \draw (24.7,-3.2) node {$\aCell{c}{2}$};
%     \draw [black] (24.5,-13.2) circle (3);
%     \draw (24.5,-13.2) node {$\aCell{c}{1}$};
%     \draw [black] (6.91,-8.74) -- (21.59,-12.46);
%     \fill [black] (21.59,-12.46) -- (20.94,-11.78) -- (20.69,-12.75);
%     \draw [black] (24.56,-10.2) -- (24.64,-6.2);
%     \fill [black] (24.64,-6.2) -- (24.12,-6.99) -- (25.12,-7.01);
%     \draw [black] (21.78,-3.88) -- (6.92,-7.32);
%     \fill [black] (6.92,-7.32) -- (7.81,-7.63) -- (7.59,-6.65);
%     \draw [black] (27.18,-11.877) arc (144:-144:2.25);
%     \fill [black] (27.18,-14.52) -- (27.53,-15.4) -- (28.12,-14.59);
%     \draw [black] (0.7,-3) -- (2.35,-5.5);
%     \fill [black] (2.35,-5.5) -- (2.32,-4.55) -- (1.49,-5.1);
%     \end{tikzpicture}
% \end{center}

\begin{center}
\begin{tikzpicture}[shorten >=1pt,node distance=2cm,/tikz/initial text =, every node/.style={scale=0.8, fill=white}]
    \tikzstyle{every state}=[]
  
    \node[state,initial]   (s)  {$s_0$};
    \node[state] (q_1) [below right=0.5cm and 0.8cm of s]  {$s_1$};
    \node[state] (q_2) [below left=0.5cm and 0.8cm of s] {$s_2$};
  
    \path[->]
    (s)   edge              node {c=0} (q_1)
    (q_1) edge [loop above]  node {c=1} (   )
            edge node {c=1} (q_2)
    (q_2) edge  node {c=2} (s);
  \end{tikzpicture}
\end{center}
  

%% file: tsl_model_checking.tex
\label{infinitestuff}

\section{Büchi Product Programs and TSL Model Checking}\label{sec:tslMC}
We now describe how we model the system and specification as Büchi automata, adapting the automata of~\cite{FairnessModTheory} to the setting of TSL. Then, we introduce our model checking algorithm for TSL(T). 
In Sec~\ref{sec:hyperMC} we build on this algorithm to propose an algorithm for HyperTSL(T) model checking. 

We use a symbolic representation of the system (see, for example,~\cite{SoftwareModelCheckingAutomata}), where transitions are labeled with program statements, and all states are accepting.

\begin{definition}\label{def:grammar} Let $c \in \Cells, \PredTerm \in \PredTerms$ and $\FuncTerm{} \in \FuncTerms$. We define the set of \emph{(basic) program statements} 
as  
\begin{align*}
    s_0 &::= \mathit{assert}(\PredTerm)~|~c:= \FuncTerm{}~|~c:=* \\
    s &::= s_0~|~s;s
\end{align*}
We call statements of the type $s_0$ \emph{basic program statements}, denoted by $\BStmt$; statements of type $s$ are denoted by $\Stmt$. 
The assignment $c:= *$ means that any value could be assigned to $c$.

%\begin{align*}
 %    &s_0 ::= assert(\PredTerm)~|~c:= \FuncTerm{}~|~c:=* \\
 %&s ~::= s_0~|~s;s
 %\end{align*}
\end{definition}

 % - for example, if $c$ is chosen as a random number. 

A \textit{program automaton} $\aut{P}$ is a 
Büchi automaton with $\Sigma = \Stmt$, that is,
$\aut{P} = (\Stmt, Q, q_0, \delta, F)$ and  $\delta \subseteq Q \times \Stmt \times Q$. 
%where $Q$ is a set of program locations; $q_0 \in Q$ is the initial location; $\delta \subseteq Q \times \Stmt \times Q$ is the transition relation; and $F \subseteq Q$ is the set of accepting states. 
When modeling the system we only need basic statements, thus we have $\Stmt = \BStmt$; and $F = Q$ as all states are accepting. See Fig.~\ref{fig:first_system_example} for  an illustration.

Using a program automaton, one can model \verb|if| statements, \verb|while| loops, and non-deterministic choices. However, not every trace of the program automaton corresponds to a program execution. For example, the trace $(n := \mathit{input}_1);\mathit{assert} (n > 0);\mathit{assert}(n < 0);~\mathit{assert}(true)^\omega$ does not -- the second assertion will always fail. Such a trace is called \textit{infeasible}. 
We call a trace \textit{feasible} if it corresponds to a program execution where all the assertions may succeed. We now define this formally.

\begin{definition} \label{def:matches}
    A computation $\Comput$ \emph{matches} a trace $\PTrace \in \BStmt^\omega$ at time point $t$, denoted by $\matchest{\Comput}{\PTrace}{t}$, if the following holds:
    \begin{align*}
        % \begin{cases}
            &\text{if } \PTrace_t = \mathit{assert}(\tau_P):&& \Eval(\PredTerm, \Comput_{t-1}) = true ~\text{ and }~ \forall c \in \Cells.~\Comput_t(c) = \Comput_{t-1}(c)  \\
            &\text{if }\PTrace_t = c:= \FuncTerm{}:&& \Eval(\FuncTerm{}, \Comput_{t-1}) = \Comput_t(c)  ~\text{ and }~ \forall c' \in \Cells \backslash \{c\}.~\Comput_t(c') = \Comput_{t-1}(c')  \\
            &\text{if } \PTrace_t = c := *: && \forall c \in \Cells\backslash\{c\}.~\Comput_t(c) = \Comput_{t-1}(c)
        % \end{cases}
    \end{align*}
    where $\Comput_{-1}$ is the initial assignment.
    A computation $\Comput$ matches a trace $\PTrace \in \BStmt^\omega$, denoted by $\matches{\Comput}{\PTrace}$, if $\forall t \in \mathbb{N}.~\matchest{\Comput}{\PTrace}{t}$.

\end{definition}

\begin{definition}
    A program automaton $\aut{P}$ over $\BStmt$ satisfies a TSL(T)-formula $\varphi$, if for all traces $\PTrace$ of $P$ we have
 $   \forall \Comput \in \Assigns^\omega.~\matches{\Comput}{\PTrace} \Rightarrow \Comput \models \varphi$.
\end{definition}

We now present an algorithm to check whether a program automaton $\aut{P}$ satisfies a TSL(T) formula. It is an adaption of the automaton-based LTL software model checking approach by~\cite{FairnessModTheory}, where the basic idea is to first translate the negated specification $\varphi$ into an automaton $\aut{A}_{\neg \varphi}$, and then combine $\aut{A}_{\neg \varphi}$ and $\aut{P}$ to a new automaton, namely the \textit{Büchi program product}. The program satisfies the specification iff the Büchi program product accepts no feasible trace.

In \cite{FairnessModTheory}, the Büchi program product is constructed similarly to the standard product automata construction. To ensure that the result is again a program automaton, the transitions are not labeled with pairs $(s, \edgel) \in \BStmt \times 2^{AP}$, but with the program statement $(s;~ \mathit{assert}(\edgel))$. A feasible accepted trace of the Büchi program product then corresponds to a counterexample proving that the program violates the specification. In the following, we discuss how we adapt the construction of the Büchi program product for TSL(T) such that this property -- a feasible trace corresponds to a counterexample -- remains true for TSL(T). %Moreover, we want to construct the TSL Büchi program product in such a way that we can use the same algorithm as in \cite{FairnessModTheory} for testing if there is a feasible accepted trace.

Let $\varphi$ be a TSL(T) specification. 
For the construction of $\aut{A}_{\neg \varphi}$, we treat all update and predicate terms as atomic propositions, resulting in an LTL formula $\neg\varphi_{\textit{LTL}}$, which is translated to a Büchi automaton.\footnote{For the translation of LTL formulas to Büchi automata, see, for example,~\cite{LTL_Buechi, LTL_Buechi2, LTL_Buechi_Tut}.} 
 For our version of the Büchi program product, we need to merge a transition label $s$ from $\aut{P}$
 with a transition label $\edgel$ from $\aut{A}_{\neg\varphi_{\textit{LTL}}}$ into a single program statement such that the assertion of the combined statement succeeds iff $\edgel$ holds for the statement $s$. Note that $\edgel$ is a set of update and predicate terms. For the update terms $\Upd{c}{\FuncTerm{}}$ we cannot just use an assertion to check if they are true, as we need to `save' the value of $\FuncTerm{}$ before the statement $s$ is executed.

Our setting differs from~\cite{FairnessModTheory} also in the fact that their program statements~do not reason over input streams. We model the behavior of input streams by using fresh memory cells that are assigned a new value at every time step.  
In the following, we define a function $\mathit{combine} $ that combines a program statement $s$ and a transition label $\edgel$ to a new program statement as described above.
%Now, we combine $A_{\neg \varphi}$ with the program automaton to create the Büchi program product, whose feasible traces correspond to feasible traces of the program automaton that do not satisfy the TSL formula. The construction is similar to that of a product automaton, but is defined in such a way that the Büchi program product is again a program automaton. To achieve that, 
\begin{definition}
Let $\USet = \{\Upd{c_1}{\FuncTerm{1}}, \dots, \Upd{c_n}{\FuncTerm{n}}\}$ be the set of update terms appearing in $\varphi$, let $\PredSet$ be the set of predicate terms appearing in $\varphi$. Let $\edgel \subseteq (\USet \cup \PredSet)$ be a transition label of $\aut{A}_{\neg \varphi}$.  Let $(tmp_j)_{j \in \mathbb{N}}$ be a family of fresh cells. Let $\Inputs = \{i_1, \dots i_m \}$. We define the function $\mathit{combine} : \Stmt \times \mathcal{P}(\PredTerms \cup \UpdTerms) \rightarrow \Stmt$ as follows. The result of $\combine{s}{\edgel}$ is composed of the program statements in $\mathit{save\_values}_\edgel, s, \mathit{new\_inputs}, \mathit{check\_preds}_\edgel$ and $\mathit{check\_updates}_\edgel$. Then we have: 
\begin{align*}
    \mathit{save\_values} &:= \mathit{tmp}_1 := \FuncTerm{1};~\dots ;\mathit{tmp}_n := \FuncTerm{n} \\
    \mathit{new\_inputs} &:= i_1 := *;~\dots~; i_m := *\\
    \mathit{check\_preds}_\edgel &:= assert \left(\bigwedge_{\PredTerm \in \edgel} \PredTerm \wedge \bigwedge_{\PredTerm \in \PredSet \backslash \edgel} \neg \PredTerm \right) \\
    \mathit{check\_updates}_\edgel &:= assert \left( \bigwedge_{\Upd{c_j}{\FuncTerm{j}} \in \USet}
    \begin{cases}
        c_j = \mathit{tmp}_j &\text{if } \Upd{c_j}{\FuncTerm{j}} \in \edgel\\
        c_j \neq \mathit{tmp}_j &\text{else} 
    \end{cases}        
        \right) \\
    \combine{s}{\edgel} &:= \mathit{save\_values};~s;~\mathit{new\_inputs};~\mathit{check\_preds}_\edgel;~\mathit{check\_updates}_\edgel
\end{align*}
\end{definition}

We can extend this definition to combining traces instead of single transition labels. 
%a program trace and a predicate trace by applying it per timepoint. 
This leads to a function $\mathit{combine} : \Stmt^\omega \times \mathcal{P}(\PredTerms \cup \UpdTerms)^\omega \rightarrow \Stmt^\omega$.
Note that the result of $\mathit{combine}$ is again a program statement in $\Stmt$ (or a trace $\Stmt^\omega$) over the new set of cells $\Cells \cup \Inputs \cup (tmp_j)_{j \in \mathbb{N}}$, which we call $\Cells^*$.

\begin{example}
    Let $\Inputs = \{i\}$. Then the result of $\combine{n := 42}{ \{ \Upd{n}{n + 7}, n > 0\}} $ is $\mathit{tmp}_0 := n + 7;~n := 42;~i := *;~ \mathit{assert} (n > 0);~ \mathit{assert} (n = \mathit{tmp}_0)$.
\end{example}

As $\mathit{combine}$ leads to composed program statements, we now need to extend the definition of feasibility to all traces. To do so, we define a function $\mathit{flatten}: \Stmt^\omega \rightarrow {\Stmt_0}^\omega$ that takes a sequence of program statements and transforms it into a sequence of basic program statements by converting a composed program statement into multiple basic program statements.

\begin{definition}
    A trace $\PTrace \in \Stmt^\omega$ \emph{matches} a computation $\Comput$, denoted by $\matches{\Comput}{\PTrace}$ if $\matches{\Comput}{\flatten{\PTrace}}$.
    A trace $\PTrace$ is \emph{feasible} if there is a computation $\Comput$ such that $\matches{\Comput}{\PTrace}$.
\end{definition}

\begin{definition}{\textbf{(Combined Product)}} 
    Let $\aut{P} = (Stmt, Q, q_0, \delta, Q)$ be a program automaton and $\aut{A} = (\mathcal{P}(\PredTerms \cup \UpdTerms), Q', q_0', \delta',F')$ be a Büchi automaton (for example, the automaton $\aut{A}_{\neg \varphi_{LTL}}$). The combined product $\aut{P} \BuechiProd \aut{A}$ is an automaton $\aut{B} = (Stmt, Q \times Q', (q_0, q_0'), \delta_B, F_{B})$, where 
    \begin{align*}
        F_{B} &= \{(q, q') \mid q \in Q \wedge q' \in F'\} \\
        \delta_B &= \{((p, q), \combine{s}{\edgel}, (p', q'))~|~(p, s, p') \in \delta \wedge (q, \edgel, q') \in \delta'\}
    \end{align*}
\begin{comment}  
\begin{itemize}
    \item $\delta_B = \{((p, q), \combine{s}{\edgel}, (p', q'))~|~(p, s, p') \in \delta \wedge (q, \edgel, q') \in \delta'\}$, and
    \item $ F_{B} = \{(q, q') \mid q \in Q \wedge q' \in F'\}$
\end{itemize}
\end{comment}
\end{definition}

\begin{theorem} \label{thm:BuechiProdModelChecking}
    Let $\aut{P}$ be a program automaton over $\BStmt$. Let $\varphi$ be a TSL(T) formula. Then $\aut{P}$ satisfies $\varphi$ if and only if $\aut{P} \BuechiProd \aut{A}_{\neg \varphi_{LTL}}$ has no feasible trace.
\end{theorem}

\begin{proof}[sketch]
    If $\matches{\Comput}{\sigma}$ is a counterexample, we can construct a computation $\tilde\Comput$ that matches the corresponding combined trace in $\aut{P} \BuechiProd \aut{A}_{\neg \varphi_{LTL}}$, and vice versa. The formal construction is given in App.~\ref{sec:buechi_corr}. 
\end{proof}

%The main idea of the proof is a construction that, given a computation that matches a program trace and violates $\varphi$, constructs a computation matching the combined trace and vice versa. For more details, see App. \ref{sec:buechi_corr}.

We can now apply Thm. \ref{thm:BuechiProdModelChecking} to solve the model checking problem by testing whether $\aut{P} \BuechiProd \aut{A}_{\neg \varphi_{LTL}}$ does not accept any feasible trace, using the feasibility check in~\cite{FairnessModTheory} as a black box. 
%After generating the combined product, we can use the algorithm of~\cite{FairnessModTheory}, which tests if a feasible trace exists. 
The algorithm of~\cite{FairnessModTheory} is based on counterexample-guided abstraction refinement (CEGAR \cite{CEGAR}). Accepted traces are checked for feasibility.  
%when a trace that is accepted by the automaton is found, the trace is checked for feasibility. 
First, finite prefixes of the trace are checked using an SMT-solver. If they are feasible, a ranking function synthesizer is used to check whether the whole trace eventually terminates. If the trace is feasible, it serves as a counterexample. If not, the automaton is refined such that it now does not include the spurious counterexample trace anymore, and the process is repeated. For more details, we refer to \cite{FairnessModTheory}. 
The limitations of SMT-solvers and ranking function synthesizers also limit the functions and predicates that can be used in both the program and in the TSL(T) formula.

\section{HyperTSL(T) Model Checking}
We now turn to the model checking problem of HyperTSL(T). We start with alternation-free formulas and continue with $\forall^*\exists^*$ formulas. 

\subsection{Alternation-free HyperTSL(T)}\label{sec:altfree}

In this section, we apply the technique of self-composition to extend the algorithm of Sec.~\ref{sec:tslMC} to alternation-free HyperTSL(T).
%, similarily to Section \ref{sec:FiniteHyper}, but now for a program automaton. 
%Self-composition is a technique commonly used for the verification of hyperproperties \cite{SelfComposition1, SelfComposition2, SelfComposition3}.
%\subsection{The Algorithm}
First, we define what it means for a program automaton to satisfy a HyperTSL(T) formula.

\begin{definition}
    Let $\aut{P}$ be a program automaton over $\BStmt$, let $\varphi$ be a HyperTSL(T) formula and let $Z=\{\Comput \in \Assigns^\omega ~|~ \exists \PTrace.~\matches{\Comput}{\PTrace} \text{ and } \PTrace \text{ is a trace of }\aut{P} \}$. 
    We say that $\aut{P}$ \emph{satisfies} $\varphi$ if $Z \models \varphi$.
\end{definition}

\begin{definition}
    Let $\aut{P} = (\Stmt, Q, q_0, \delta, Q)$ be a program automaton. The \emph{$n$-fold self-composition} of $\aut{P}$ is $\aut{P}^n = (\Stmt', Q^n, q_0^n, \delta^n, Q^n)$, where $\Stmt'$ are program statements over the set of inputs $\Inputs \times \TraceVs$ and the set of cells $\Cells \times \TraceVs$ and where $Q^n = Q \times \dots \times Q$,  $q_0^n = (q_0, \dots , q_0)$ and
    \begin{align*}
        \delta^n = &\{((q_1, \dots, q_n), ((s_1)_{\pi_1}; \dots; (s_n)_{\pi_n}), (q_1', \dots, q_n')) \\
        & \quad \mid \forall 1 \leq i \leq n.~ (q_i, s_i, q_i') \in~\delta\}
    \end{align*}
where $(s)_{\pi}$ renames every cell $c$ used in $s$ to $c_{\pi}$ and every input $i$ to $i_{\pi}$.
\end{definition}

\begin{theorem} \label{thm:BuechiProdModelCheckingAFH}
    A program automaton $\aut{P}$ over $\BStmt$ satisfies a universal HyperTSL(T) formula $\varphi = \forall \pi_1.~ \dots \forall \pi_n.~\psi$ iff $\aut{P}^n \BuechiProd \aut{A}_{\neg \psi_{LTL}}$ has no feasible trace.
\end{theorem}

\begin{theorem} \label{thm:BuechiProdModelCheckingAFH2}
    A program automaton $P$ over $\BStmt$ satisfies an existential HyperTSL(T) formula $\varphi = \exists \pi_1.~ \dots \exists \pi_n.~\psi$ iff $\aut{P}^n \BuechiProd \aut{A}_{\psi_{LTL}}$ has some feasible trace.
\end{theorem}

\noindent The proofs of are analogous to the proof of Thm.~\ref{thm:BuechiProdModelChecking} and are provided in~App.~\ref{sec:BuechiProd_corr2}.

%% file: infinite_forall_exists.tex
\subsection{ $\forall^* \exists^*$ HyperTSL(T)}\label{sec:hyperMC}

In this section, we present a sound but necessarily incomplete algorithm for finding counterexamples for $\forall^* \exists^*$ HyperTSL(T) formulas.\footnote{Note that the algorithms of Sec.~\ref{sec:tslMC} and Sec.~\ref{sec:altfree} are also incomplete, due to the feasibility test. However, the incompleteness of the algorithm we provide in this section is inherent to the quantifier alternation of the formula.}
%there is also for HyperLTL no such algorithm yet, thus this also gives the first software model checking algorithm for finding counterexamples for $\forall^* \exists^*$ HyperLTL formulas. 
Such an algorithm %that finds counterexamples for the $\forall^* \exists^*$ fragment 
can also provide witnesses $\exists^* \forall^*$ formulas. As HyperTSL(T) is built on top of HyperLTL, we combine ideas from finite-state HyperLTL model checking~\cite{HyperLTLModelChecking}
with the algorithms of Sec.~\ref{sec:tslMC} and Sec.~\ref{sec:altfree}.

Let $\varphi = \forall^m \exists^n. \psi$. For HyperLTL model checking, \cite{HyperLTLModelChecking} first constructs an automaton containing the system traces satisfying $\psi_\exists := \exists ^n. \psi$, and then applies complementation to extract counterexamples for the $\forall\exists$ specification.
Consider the automaton $\aut{P}^n \BuechiProd \aut{A}_{\psi_{LTL}}$ from Sec.~\ref{sec:tslMC}, whose feasible traces correspond to the system traces satisfying $\psi_\exists$. If we would be able to remove all infeasible traces, we could apply the finite-state HyperLTL model checking construction.
Unfortunately, removing all infeasibilities is impossible in general, as the result would be a finite-state system describing exactly an infinite-state system. Therefore, the main idea of this section is to remove parts of the infeasible traces from $\aut{P}^n \BuechiProd \aut{A}_{\psi_{LTL}}$, constructing an over-approximation of the system traces satisfying $\psi_\exists$. 
A counterexample disproving $\varphi$ is then a combination of system traces that is not contained in the over-approximation. 

We propose two techniques for removing infeasibility. The first technique removes \textit{k-infeasibility} from the automaton, that is, a local inconsistency in a trace, occurring within $k$ consecutive time steps. When choosing $k$, there is a tradeoff: if $k$ is larger, more counterexamples can be identified, but the automaton construction gets exponentially larger. 

The second technique removes \textit{infeasible accepting cycles} from the automaton. It might not be possible to remove all of them, thus we bound the number of iterations. We present an example and then elaborate on these two methods.

\begin{example}
    The trace $t_1$ below is 3-infeasible, because regardless of the value of $n$ prior to the second time step, the assertion in the fourth time step will fail.
    %\begin{align*}
     $$   t_1 = (n--;~ \mathit{assert}(n >= 0))~ (n:= 1;~ \mathit{assert} (n >= 0))~ (n--;~ \mathit{assert} (n >= 0))^\omega $$
    %\end{align*}
    In contrast, the trace 
    $t_2 = (n := *)~(n--;~ \mathit{assert} (n >= 0))^\omega$
    is not $k$-infeasible for any $k$, because the value of $n$ can always be large enough to pass the first $k$ assertions. Still, the trace is infeasible because $n$ cannot decrease forever without dropping below zero. If such a trace is accepted by an automaton, $n--;~\mathit{assert} (n >= 0)$ corresponds to an infeasible accepting cycle.    
\end{example}

\subsubsection{Removing $k$-infeasibility} 
 To remove $k$-infeasibility from an automaton, we  construct a new program automaton that `remembers' the $k-1$ previous statements. The states of the new automaton correspond to paths of length $k$ in the original automaton. We add a transition labeled with $l$ between two states $p$ and $q$ if we can extend the trace represented by $p$ with $l$ such that the resulting trace is $k$-feasible. Formally, we get:

\begin{definition}
    Let $k \in \mathbb{N}$, $\PTrace \in \Stmt^\omega$. We say that $\PTrace$ is \emph{$k$-infeasible} if there exists $j \in \mathbb{N}$ such that $\PTrace_j \PTrace_{j+1} \dots \PTrace_{j+k-1}; \mathit{assert}(true)^\omega$ is infeasible for all possible initial assignments $\Comput_{-1}$. We then also call the subsequence $\PTrace_{j}\PTrace_{j+1}\dots \PTrace_{j+k-1}$ infeasible.
    If a trace is not $k$-infeasible, we call it $k$-feasible.\footnote{Whether a subsequence $\PTrace_j\PTrace_{j+1} \dots \PTrace_{j+k-1}$ is a witness of k-infeasibility can be checked using an SMT-solver, e.g, \cite{Z3, SMT1, SMT2, SMT3}.}
\end{definition}

\begin{definition}
    Let $\aut{P} = (\Stmt, Q, q_0, \delta, F)$ be a program automaton. Let $k \in \mathbb{N}$. We define $\aut{P}$ without $k$-infeasibility, as $\aut{P}_k = (\Stmt, Q', q_0, \delta', F')$ where
    \begin{align*}
        Q' :=& \{(q_1,s_1,q_2 \dots ,s_{k-1},q_k) \mid (q_1, s_1, q_2) \in \delta \wedge \dots \wedge (q_{k-1}, s_{k-1}, q_k) \in \delta \}~ \cup \\
        &\{(q_0, s_0,q_1 \dots ,s_{k' - 1}, q_{k'}) \mid k' < k-1 \wedge (q_0, s_0, q_1) \in \delta \wedge \dots \\
        & \phantom{(q_0, s_0,q_1 \dots ,s_{k' - 1}, q_{k'}) \mid} \quad \wedge (q_{k'-1}, s_{k'-1}, q_{k'}) \in \delta  \} \\
        \delta' :=& \{((q_1,s_1,q_2\dots ,s_{k-1},q_k), s_k, (q_2,s_2, \dots ,q_k,s_k,q_{k+1})) \in Q' \times \Stmt \times Q' \\ &\quad \mid s_1 \dots s_k \text{ feasible} \}~\cup \\
        &\{((q_0,s_0,q_1\dots ,s_{k'-1},q_{k'}), s_{k'}, (q_0, s_0, \dots ,q_{k'},s_{k'},q_{k'+1})) \in Q' \times \Stmt \times Q' \\ &\quad \mid k' < k-1 \wedge s_0 \dots s_{k'} \text{ feasible} \} \\
        F' :=& \{(q_1,s_1,q_2 \dots ,s_{k-1},q_k) \in Q' \mid q_k \in F \}~\cup \\
        & \{(q_0, s_0,q_1 \dots ,s_{k' - 1}, q_{k'}) \in Q' \mid k' < k-1 \wedge q_{k'} \in F\}
    \end{align*}
\end{definition}

\begin{theorem} \label{thm:k_feasible}
   $\aut{P}_k$ accepts exactly the $k$-feasible traces of~$\aut{P}$.
\end{theorem}
The proof follows directly from the construction above, see App.~\ref{sec:k_feasible_proof} for details. 

\subsubsection{Removing Infeasible Accepting Cycles} 
For removing infeasible accepting cycles, we first enumerate all simple cycles of the automaton (using, e.g.,~\cite{CycleFinding}), adding also cycles induced by self-loops. For each cycle $\varrho$ that contains at least one accepting state, we test its feasibility: first, using an SMT-solver to test if $\varrho$ is locally infeasible; then, using a ranking function synthesizer (e.g., \cite{Termination1, Termination2, Termination3}) to test if $\varrho^\omega$ is infeasible. If we successfully prove infeasibility, we refine the model, using the methods from \cite{SoftwareModelCheckingAutomata, TerminationRefinement}. This refinement is formalized in the following.

\begin{definition}
    Let $\aut{P}=(\Stmt, Q, q_0, \delta, F)$ be a program automaton. Let $\varrho = (q_1, s_1, q_2)(q_2, s_2, q_3)\dots(q_n, s_n, q_1)$ be a sequence of transitions of $\aut{P}$. We say that~$\varrho$ is an \emph{infeasible accepting cycle} if there is a $1 \leq j \leq n$ with $q_j \in F$ and $(s_1 s_2 \dots s_{n-1})^\omega$ is infeasible for all possible initial assignments $\Comput_{-1}$.
\end{definition}

\begin{definition}
    Let $\aut{P}$ be a program automaton and $C \subseteq (Q \times \Stmt \times Q)^\omega$ be a set of infeasible accepting cycles of $\aut{P}$.
    Furthermore, let
    $$\varrho = (q_1, s_1, q_2)(q_2, s_2, q_3)\dots(q_{n-1}, s_{n-1}, q_n) \in~C.$$
    The automaton $\aut{A}_\varrho$ for $\varrho$ is $ \aut{A}_\varrho = (\Stmt, Q=\{q_0, q_1, \dots q_n\}, q_0, \delta, Q \backslash \{q_0 \}) $ where 
   \begin{align*}
        \delta ~=~ &\{(q_0, s, q_0) \mid s \in \Stmt \} \\
        & \cup \{(q_j, s_j, q_{j+1}) \mid 1 \leq j < n \} \cup \{(q_0, s_1, q_2), (q_n, s_n, q_1)\}.
   \end{align*}
    \end{definition}
   Then, $\aut{A}_\varrho$ accepts exactly the traces that end with $\varrho^\omega$, without any restriction on the prefix. See Fig. \ref{fig:aut_infeasible_cycle} for an example. To exclude the traces of $\aut{A}_\varrho$ from $\aut{P}$, we define
    $
        \aut{P}_C := \aut{P} \backslash \left( \bigcup_{\varrho \in C} \aut{A}_\varrho \right)
    $.\footnote{For two automata $\aut{A}_1, \aut{A}_2$ we use $\aut{A}_1 \backslash \aut{A}_2$ to denote the intersection of $\aut{A}_1$ with the complement of $\aut{A}_2$, resulting in the language $\aut{L}(\aut{A}_1) \setminus\aut{L}(\aut{A}_2)$. }
    This construction can be repeated to exclude infeasible accepted cycles that are newly created in $\aut{P}_C$. We denote the result of iterating this process $k'$ times by~$\aut{P}_{C(k')}$.

\subsubsection{Finding Counterexamples for $\forall^* \exists^*$ HyperTSL(T)-Formulas} 
Consider now a HyperTSL(T) formula $\varphi = \forall^{1\cdots m}\exists ^{m+1\cdots n}.\psi$ and a program automaton $\aut{P}$.
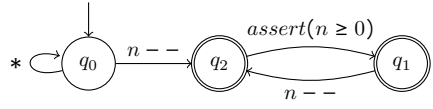
\begin{wrapfigure}{r}{0.5\textwidth}
    \vspace{-1cm}
    \input{aut_infeasible_cycle.tex}
    \vspace{-5mm}
    \caption{Automaton $\aut{A}_{\varrho}$ for the infeasible cycle
    $\varrho = (q_1,~n--,~q_2)(q_2,~ assert(n>0),~ q_1)$. Label $*$ denotes an edge for every (relevant) statement.}
    \label{fig:aut_infeasible_cycle}
    \vspace{-1cm}
\end{wrapfigure}
For finding a counterexample, we first construct the combined product $\aut{P}^n \BuechiProd \aut{A}_\psi$. 
Each feasible accepted trace of $\aut{P}^n \BuechiProd \aut{A}_\psi$ corresponds to a combination of $n$ feasible program traces that satisfy $\psi$. Next, we eliminate $k$-infeasibility and remove $k'$-times infeasible accepting cycles from the combined product, resulting in the automaton $(\aut{P}^n \BuechiProd \aut{A}_\psi)_{k, C(k')}$. Using this modified combined product, we obtain an over-approximation of the program execution combinations satisfying the existential part of the specification.
 Each trace of the combined product is a combination of $n$ program executions and a predicate/update term sequence. We then project the $m$ universally quantified program executions from a feasible trace, obtaining a tuple of $m$ program executions that satisfy the existential part of the formula. Applying this projection to all traces of $(\aut{P}^n \BuechiProd \aut{A}_\psi)_{k, C(k')}$ leads to an over-approximation of the program executions satisfying the existential part of the specification. Formally:

\begin{definition}
    Let $\aut{P}$ be a program automaton, let $m\leq n \in \mathbb{N}$, and let $\aut{A}_\psi$ be the automaton for the formula $\psi$. Let $(\aut{P}^n \BuechiProd \aut{A})_{k, C(k')} = (\Stmt, Q, q_0, \delta, F)$. We define the \emph{projected automaton} $(\aut{P}^m \BuechiProd \aut{A})_{k, C(k')}^\forall = (\Stmt, Q, q_0, \delta^\forall, F)$ where
    %\begin{align*}
 $   \delta^\forall = \{(q, (s_1; \dots ;~s_m), q') \mid \exists s_{m+1}, \dots s_n, \edgel.~ (q, \combine{s_1; \dots ;~s_n}{\edgel},q') \in \delta \} $.
 %\footnote{Recall that $s_1;s_2$ refers to a sequence of statements, as given in Def.~\ref{def:grammar}. For more details on the universal projection we refer the reader to\cite{DBLP:conf/fsttcs/FinkbeinerP22}.}
    %\end{align*}
\end{definition}
The notation $s_1;s_2$ refers to a sequence of statements, as given in Def.~\ref{def:grammar}. For more details on the universal projection we refer the reader to\cite{DBLP:conf/fsttcs/FinkbeinerP22}.

Now, it only remains to check whether the over-approximation contains all tuples of $m$ feasible program executions. If not, a counterexample is found. This boils down to testing if $\aut{P}^m \backslash (\aut{P}^n \BuechiProd \aut{A}_\psi)_{k, C(k')}^\forall$ has some feasible trace.
 Thm.~\ref{thm:corr_AE} states the soundness of our algorithm. See App.~\ref{sec:corr_AE} for its proof. 

\begin{theorem} \label{thm:corr_AE}
    Let $\varphi = \forall^{1\cdots m}\exists ^{m+1\cdots n}. \psi $ be a HyperTSL(T) formula. If the automaton $\aut{P}^m \backslash (\aut{P}^n \BuechiProd \aut{A}_\psi)_{k, C(k')}^\forall$ has a feasible trace, then $\aut{P}$ does not satisfy $\varphi$.
\end{theorem}

\section{Demonstration of the Algorithm} \label{sec:HyperExamples}

In this section, we apply the algorithm of Sec.~\ref{sec:hyperMC} to two simple examples, demonstrating that removing some infeasibilities can already be sufficient for identifying counterexamples.

\subsubsection{Generalized Noninterference}
Recall the formula $\varphi_{gni} = \forall \pi.~\exists \pi'.~\LTLglobally(i_{\pi'} = \lambda \wedge c_\pi = c_{\pi'})$ introduced in Sec.~\ref{sec:intro}, specifying generalized noninterference. 
We model-check $\varphi_{gni}$ on the program automaton $\aut{P}$ of Fig.~\ref{fig:gni_example_2} (left), setting $\lambda = 0$.
The program $\aut{P}$ violates $\varphi_{gni}$ since for the trace $(assert (i < 0)~ c:=0)^\omega$ there is no other trace where on which $c$ is equal, but $i=0$.
%$P^2$ is shown in Fig.~\ref{fig:gni_example_2}. 

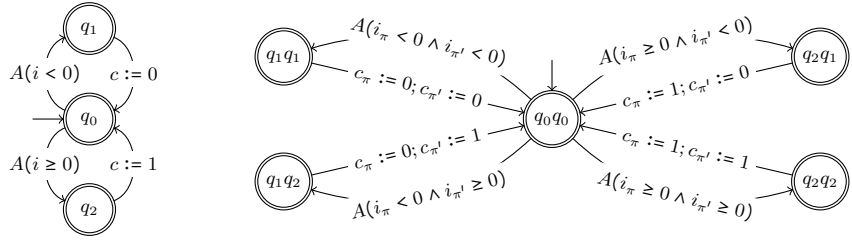
\begin{figure}[t]
\begin{minipage}{0.25\textwidth}
\input{gni_example_1.tex}
\end{minipage}%
\hfill
\begin{minipage}{0.74\textwidth}
\input{gni_example_2.tex}
\end{minipage}
\caption{Left: The program automaton $\aut{P}$ used in the first example. Right: The program automaton $\aut{P}^2$. For brevity, we use $A$ for $assert$ and join consecutive assertions. }
\label{fig:gni_example_2}
\end{figure}

The automaton for $\psi = \LTLglobally(i_{\pi'} = 0 \wedge c_\pi = c_{\pi'})$ consists of a single accepting state with the self-loop labeled with $\PredTerm = (i_{\pi'} = 0 \wedge c_{\pi} = c_{\pi'})$. For this example, it suffices to choose $k=1$. To detect $1$-inconsistencies we construct $\aut{P}^2$ (Fig~\ref{fig:gni_example_2}, right). Then, $(\aut{P}^2 \BuechiProd \aut{A}_{\psi})_k$ is the combined product with all
 $1$-inconsistent 
 \begin{wrapfigure}{r}{0.5\textwidth}
\vspace{-3mm}
\input{gni_example_4.tex}
\caption{program automaton $(\aut{P}^2 \BuechiProd \aut{A}_{\psi})_k^\forall$}
\label{fig:gni_example_4}
\vspace{-3mm}
\end{wrapfigure}
 transitions removed (see Fig.~\ref{fig:gni_example_3} for the combined product).

\begin{figure}[t]
\input{gni_example_3.tex}
\caption{The combined product $(\aut{P}^2 \BuechiProd \aut{A}_{\psi})$}
\label{fig:gni_example_3}
\end{figure}

The automaton $(\aut{P}^2 \BuechiProd \aut{A}_{\psi})_k^\forall$ is shown in Fig.~\ref{fig:gni_example_4}. 
It does not contain the trace $\sigma = \mathit{assert}(i < 0)~(c:=0)^\omega$ which is a feasible trace of $\aut{P}$. Therefore, $\sigma$ is a feasible trace accepted by $\aut{P}\backslash (\aut{P}^2 \BuechiProd \aut{A}_{\psi})_k^\forall$ and is a counterexample proving that $\aut{P}$ does not satisfy generalized noninterference -- there is no feasible trace that agrees on the value of the cell $c$ but has always $i=0$. %For this example, it is not necessary to remove infeasible cycles.

\subsubsection{The Need of Removing Cycles} 

We now present an example in which removing $k$-infeasibility is not sufficient, but removing infeasible accepting cycles leads to a counterexample. Consider the specification
$
    \varphi = \forall \pi\exists \pi'.\LTLglobally (p_\pi \neq p_{\pi'} \wedge n_\pi < n_{\pi'})
$
and the program automaton $\aut{P}_{cy}$ of Fig. \ref{fig:cycle_example_1}.
The formula $\varphi$ states that for every trace $\pi$, there is another trace $\pi'$ which differs from $\pi$ on $p$, but in which $n$ is always greater. The trace $\pi = (n:= *); (p := *); \mathit{assert}(p = 0); (n--)^\omega$ is a counterexample for $\varphi$ in $\aut{P}_{cy}$  as any trace $\pi'$ which differs on $p$ will decrease its $n$ by $2$ in every time step, and thus $n_{\pi'}$ will eventually drop below $n_\pi$.

\begin{figure}[t]
\begin{minipage}{0.19\textwidth}
\input{cycle_example_1.tex}
\end{minipage}%
\hfill
\begin{minipage}{0.79\textwidth}
\input{cycle_example_2.tex}
\end{minipage}
\caption{Left: The program automaton $\aut{P}_{cy}$, Right: The program automaton $\aut{P}_{cy}^2$.}
\label{fig:cycle_example_1}
\end{figure}
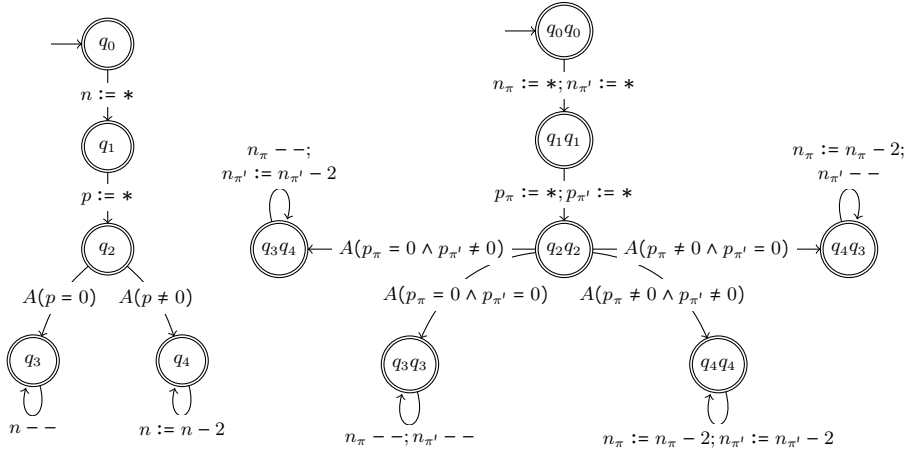

The automaton 
$\aut{P}_{cy}^2$ is shown in Fig. \ref{fig:cycle_example_1}.
In the combined product, the structure of the automaton stays the same, and $\mathit{assert}(p_\pi \neq p_{\pi'} \wedge n_\pi < n_\pi')$ is added to every state. 
Removing local $k$-infeasibilities is not sufficient here; assume $k =1$. The only $1$-infeasible transition is the transition from $q_2q_2$ to $q_3q_3$, and this does not eliminate the counterexample $\pi$. Greater $k$'s do not work as well, as the remaining traces of the combined product are not $k$~infeasible for any~$k$.

However, the self-loop at $q_3q_4$ is an infeasible accepting cycle -- the sequence \\ $(n_\pi--;~ n_{\pi'} := n_{\pi'} - 2;~ \mathit{assert} (n_{\pi} < n_{\pi'}))^\omega$ must eventually terminate. We choose $k'=1$ removing all traces ending with this cycle. Next, we project the automaton to the universal part. The trace $\pi$ is not accepted by the automaton $(\aut{P}^2 \BuechiProd \aut{A}_\psi)^\forall_{1, C(1)}$. But since $\pi$ is in $\aut{P}$ and feasible, it is identified as a counterexample.

%% file: aut_infeasible_cycle.tex
\begin{center}
    \begin{tikzpicture}[node distance= 3cm, /tikz/initial text =, every node/.style={scale=0.8}]
        \node[state, initial above] (s0) {$q_0$};
        \node[state, accepting] (s2) [right=1cm of s0] {$q_2$};
        \node[state, accepting] (s1) [right of=s2] {$q_1$};

        \path[->]
        (s0) edge [loop left] node {$*$} ( )
            edge [above] node {$n--$} (s2)
        (s1) edge [bend left = 15, below] node {$n--$} (s2)
        (s2) edge [bend left = 15, above] node {$assert(n \geq 0)$} (s1);
    
    \end{tikzpicture}
    \end{center}

%% file: gni_example_1.tex
\begin{center}
\begin{tikzpicture}[node distance = 1.5cm,  /tikz/initial text =,every node/.style={scale=0.8}]
    \node[state, initial left, accepting] (s0) {$q_0$};
    \node[state, accepting] (s1) [above of=s] {$q_1$};
    \node[state, accepting] (s2) [below of=s] {$q_2$};

    \path[->]
    (s0) edge [bend left=70] node [fill=white] {$A(i < 0)$} (s1)
        edge [bend right=70] node  [fill=white] {$A(i \geq 0)$} (s2)
    (s1) edge [bend left=70] node [fill=white] {$c := 0$} (s0)
    (s2) edge [bend right=70] node [fill=white] {$c := 1$} (s0);
    
\end{tikzpicture}
\end{center}

%% file: gni_example_2.tex
\begin{center}
\begin{tikzpicture}[ node distance = 4cm, /tikz/initial text =, every node/.style={scale=0.8, fill=white} ]
    \node[state, initial above, accepting] (s00) {$q_0q_0$};
    \node[state, accepting] (s11) [above left=0.3cm and 3cm of s00] {$q_1q_1$};
    \node[state, accepting] (s12) [below left=0.3cm and 3cm of s00] {$q_1q_2$};
    \node[state, accepting] (s21) [above right=0.3cm and 3cm of s00] {$q_2q_1$};
    \node[state, accepting] (s22) [below right=0.3cm and 3cm of s00] {$q_2q_2$};

    \path[->]
    (s00) edge [bend right=30,sloped] node {$A(i_\pi < 0 \wedge i_{\pi'} < 0)$} (s11)
        edge [bend right=30, sloped] node {$A(i_\pi \geq 0 \wedge i_{\pi'} \geq 0)$} (s22)
        edge [bend left=30, sloped] node {$A(i_\pi \geq 0 \wedge i_{\pi'} < 0 )$} (s21)
        edge [bend left =30, sloped] node {$A(i_\pi < 0 \wedge i_{\pi'} \geq 0)$} (s12)
    (s11) edge [sloped] node {$c_\pi := 0; c_{\pi'} := 0$} (s00)
    (s22) edge [sloped] node {$c_\pi := 1; c_{\pi'} := 1$} (s00)
    (s21) edge [sloped] node {$c_\pi := 1; c_{\pi'} := 0$} (s00)
    (s12) edge [sloped] node {$c_{\pi} := 0; c_{\pi'} := 1$} (s00);
    
\end{tikzpicture}
\end{center}

%% file: gni_example_4.tex
\begin{center}
\begin{tikzpicture}[node distance = 5cm, /tikz/initial text = ,every node/.style={scale=0.8, fill=white}]
    \node[state, initial above, accepting] (s00) {$q_0q_0$};
    \node[state, accepting] (s11) [above left=0cm and 2cm of s00] {$q_1q_1$};
    \node[state, accepting] (s12) [below left=0cm and 2cm of s00] {$q_1q_2$};
    \node[state, accepting] (s21) [above right=0cm and 2cm of s00] {$q_2q_1$};
    \node[state, accepting] (s22) [below right=0cm and 2cm of s00] {$q_2q_2$};

    \path[->]
    (s00) edge [bend right=15, sloped] node {$assert(i\geq 0)$} (s22)
        edge [bend left =15, sloped] node {$assert(i < 0)$} (s12)
    (s11) edge [bend left=15, sloped] node {$c:= 0$} (s00)
    (s22) edge [bend right=15, sloped] node {$c:= 1$} (s00);
    
\end{tikzpicture}
\end{center}

%% file: gni_example_3.tex
\begin{center}
    \tikzset{ every node/.style={scale=0.8, fill=white}}
\begin{tikzpicture}[node distance = 5cm, /tikz/initial text = ]
    \node[state, initial above, accepting] (s00) {$q_0q_0$};
    \node[state, accepting] (s11) [above left=0.3cm and 5cm of s00] {$q_1q_1$};
    \node[state, accepting] (s12) [below left=0.3cm and 5cm of s00] {$q_1q_2$};
    \node[state, accepting] (s21) [above right=0.3cm and 5cm of s00] {$q_2q_1$};
    \node[state, accepting] (s22) [below right=0.3cm and 5cm of s00] {$q_2q_2$};

    \path[->]
    (s00) edge [bend right=30, sloped] node {$i := *; A(i_\pi < 0 \wedge i_{\pi'} < 0 \wedge \PredTerm)$} (s11)
        edge [bend right=30, sloped] node {$i := *; A(i_\pi \geq 0 \wedge i_{\pi'} \geq 0 \wedge \PredTerm)$} (s22)
        edge [bend left=30, sloped] node {$i := *;A(i_\pi \geq 0 \wedge i_{\pi'} < 0 \wedge \PredTerm)$} (s21)
        edge [bend left =30, sloped] node {$i := *;A(i_\pi < 0 \wedge i_{\pi'} \geq 0 \wedge \PredTerm)$} (s12)
    (s11) edge [sloped] node {$i := *; c_\pi := 0; c_{\pi'} := 0; A(\PredTerm)$} (s00)
    (s22) edge [sloped] node {$i := *; c_\pi := 1; c_{\pi'} := 1; A(\PredTerm)$} (s00)
    (s21) edge [sloped] node {$i := *;c_\pi := 1; c_{\pi'} := 0; A(\PredTerm)$} (s00)
    (s12) edge [sloped] node {$i := *;c_{\pi} := 0; c_{\pi'} := 1; A(\PredTerm)$} (s00);
    
\end{tikzpicture}
\end{center}

%% file: cycle_example_1.tex
\begin{center}
\begin{tikzpicture}[node distance = 1.7cm, /tikz/initial text =, every node/.style={scale=0.8, fill=white}]
    \node[state, initial, accepting] (s0) {$q_0$};
    \node[state, accepting] (s1) [below of=s0] {$q_1$};
    \node[state, accepting] (s2) [below of=s1] {$q_2$};
    \node[state, accepting] (s3) [below left= 1cm and 0.5cm of s2] {$q_3$};
    \node[state, accepting] (s4) [below right=1cm and 0.5cm of s2] {$q_4$};

    \path[->]
    (s0) edge node {$n := *$} (s1)
    (s1) edge node {$p := *$} (s2)
    (s2) edge [bend right=15] node {$A(p = 0)$} (s3)
    (s2) edge [bend left=15] node {$A(p \neq 0)$}  (s4)
    (s3) edge [loop below] node {$n--$} ( )
    (s4) edge [loop below] node {$n := n-2$} ( );
    
\end{tikzpicture}
\end{center}

%% file: cycle_example_2.tex
\begin{center}
\begin{tikzpicture}[node distance = 3cm, /tikz/initial text =, every node/.style={scale=0.8, fill=white}]
    \node[state, initial, accepting] (s0s0) {$q_0q_0$};
    \node[state, accepting] (s1s1) [below =0.7cm of s0s0] {$q_1q_1$};
    \node[state, accepting] (s2s2) [below =0.7cm of s1s1] {$q_2q_2$};
    \node[state, accepting] (s3s4) [left =3cm of s2s2] {$q_3q_4$};
    \node[state, accepting] (s4s3) [right =3cm of s2s2] {$q_4q_3$};
    \node[state, accepting] (s3s3) [below left = 1cm and 1.5cm of s2s2] {$q_3q_3$};
    \node[state, accepting] (s4s4) [below right = 1cm and 1.5cm of s2s2] {$q_4q_4$};

    \path[->]
    (s0s0) edge node {$n_{\pi} := *; n_{\pi'} := *$} (s1s1)
    (s1s1) edge node {$p_{\pi} := *; p_{\pi'} := *$} (s2s2)
    (s2s2) edge [bend right=30] node [below] {$A(p_\pi = 0 \wedge p_{\pi'} = 0)$} (s3s3)
    (s2s2) edge [bend left=30] node [below] {$A(p_\pi \neq 0 \wedge p_{\pi'} \neq 0)$}  (s4s4)
    (s2s2) edge node {$A(p_\pi = 0 \wedge p_{\pi'} \neq 0)$} (s3s4)
    (s2s2) edge node {$A(p_\pi \neq 0 \wedge p_{\pi'} = 0 )$} (s4s3)
    (s3s3) edge [loop below] node {$n_\pi--; n_{\pi'}--$} ( )
    (s4s4) edge [loop below] node {$n_\pi := n_\pi-2; n_{\pi'} := n_{\pi'}-2$} ( )
    (s4s3) edge [loop above] node[align=center] {$n_\pi := n_\pi-2;$ \\ $n_{\pi'}--$} ( )
    (s3s4) edge [loop above] node[align=center] {$n_\pi--;$ \\ $n_{\pi'} := n_{\pi'}-2$} ( );

\end{tikzpicture}
\end{center}

%% file: discussion.tex
	\section{Conclusions}
We have extended HyperTSL with theories, resulting in HyperTSL(T), and provided the first infinite-state model checking algorithms for both TSL(T) and HyperTSL(T). As this is the first work to study (Hyper)TSL model checking, these are also the first algorithms for \emph{finite-state} model checking for (Hyper)TSL. For TSL(T), we have adapted known software model checking algorithm for LTL to the setting of TSL(T). We then used the technique of self-composition to generalize this algorithm to the alternation-free fragment of HyperTSL(T). 

We have furthermore described a sound but necessarily incomplete algorithm for finding counterexamples for $\forall^*\exists^*$-HyperTSL(T) formulas (and witnesses proving $\exists^* \forall^*$ formulas). Our algorithm makes it possible to find program executions violating properties like generalized noninterference, which is only expressible by using a combination of universal and existential quantifiers.

  Finding model checking algorithms for other fragments of HyperTSL(T), and implementing our approach, remains as future work.

%% file: correctnessproof.tex
\subsection{Hyper Linear Temporal Logic (HyperLTL)}
Let AP be a finite set of atomic propositions and let $\Pi$ be a finite set of trace variables. Then, a {HyperLTL formula} is defined by the grammar
    \begin{align*}
        \varphi ::= a_\pi \mid \neg \varphi \mid \varphi \wedge \varphi \mid \LTLnext \varphi \mid \varphi \LTLuntil \varphi \mid \forall \pi.~ \varphi \mid \exists \pi.~ \varphi
    \end{align*}

\noindent where $a \in AP$ and $\pi \in \Pi$. 

The satisfaction of a HyperLTL formula is defined with respect to a mapping $m: \Pi \rightarrow (2^{AP})^\omega$ of trace variables to traces. For treating the quantifiers, we need the notion of extending such a mapping for a new trace variable. We define
\begin{align*}
    m[\pi \rightarrow s] (\pi) &= s \\
    m[\pi \rightarrow s] (\pi') &= m(\pi') &&\text{for } \pi \neq \pi'
\end{align*}

    The satisfaction of a HyperLTL formula with respect to a set of traces $Z \subseteq {2^{AP}}^\omega$, a mapping of trace variables to traces $m: \TraceVs \rightarrow {2^{AP}}^\omega$ and a time point $t$ is recursively defined by 
    \begin{align*}
        &Z, t, m \lmodels a_{\pi} &&\Leftrightarrow a \in m(\pi)_t \\
        &Z, t, m \lmodels \neg \varphi &&\Leftrightarrow \neg (Z, t, m \lmodels \varphi) \\
        &Z, t, m \lmodels \varphi \wedge \psi &&\Leftrightarrow Z, t, m \lmodels \varphi \wedge Z, t, m \models \psi \\
        &Z, t, m \lmodels \LTLnext \varphi &&\Leftrightarrow Z, t + 1, m \lmodels \varphi \\
        &Z, t, m \lmodels \varphi\LTLuntil\psi &&\Leftrightarrow \exists t' \geq t.~ Z, t', m \lmodels \psi \wedge \forall t \leq t'' < t'.~ Z, t'', m \lmodels \varphi \\
        &Z, t, m \lmodels \forall \pi.~\varphi &&\Leftrightarrow \forall s \in Z.~m[\pi \rightarrow s] \lmodels \varphi \\
        &Z, t, m \lmodels \exists \pi.~\varphi &&\Leftrightarrow \exists s \in Z.~m[\pi \rightarrow s] \lmodels \varphi
    \end{align*}
    We define $Z \lmodels \varphi$ as $Z, 0, \emptyset \lmodels \varphi$.

For the sepcial case in which there is only one trace quantifier, and this is a universal quantifier, we are in the fragment of LTL. 
\subsection{Similiarity of LTL and TSL}

The following lemma states an important relation between (Hyper)TSL and LTL. The LTL semantics is defined with respect to a sequence of subsets of atomic propositions, while the semantics of a TSL-formula or quantifier-free HyperTSL formula is defined with respect to a (hyper-)computation. A crucial observation for this thesis is that we can `translate' between the two -- a (hyper-)computation defines a sequence of predicate and update term subsets. For each time point, the subset contains exactly the predicate and update terms that are true now.

\begin{definition} \label{def:TSL_LTL}
    Let $\HComput \in \HAssigns^\omega, \PredSet \subseteq \HPredTerms, \USet \subseteq \HUpdTerms$. We define
    \begin{align*}
        \UPredSeq(\HComput, \PredSet, \USet)_t &= \{\HPredTerm \in \PredSet \mid t, \emptyset, \HComput \models \HPredTerm\} \cup \{\Upd{c}{\HFuncTerm{}} \in \USet \mid t, \emptyset, \HComput \models \Upd{c}{\FuncTerm{}}\} \\
        \UPredSeq(\HComput, \PredSet, \USet) &= \UPredSeq(\HComput, \PredSet, \USet)_0 ~\UPredSeq(\HComput, \PredSet, \USet)_1 ~\UPredSeq(\HComput, \PredSet, \USet)_2 \dots
    \end{align*}

    If $\PredSet$ and $\USet$ are clear from the context, we also omit these arguments.
\end{definition}

\begin{lemma} \label{lem:TSL_LTL}
    Let $t \in \mathbb{N}$. Let $\varphi$ be a HyperTSL-formula without quantifiers. Let $\PredSet \subseteq \HPredTerms, \USet \subseteq \HUpdTerms$ be the sets of predicate and update terms appearing in $\varphi$, respectively. Then
    \begin{align*}
        t, \UPredSeq(\HComput) \lmodels \varphi \Leftrightarrow t, \emptyset, \HComput \models \varphi
    \end{align*}
\end{lemma}
\begin{proof} (Lemma \ref{lem:TSL_LTL})
    Proof by structural induction over $\varphi$.
    \begin{itemize}
        \item Case $\varphi = \HPredTerm$
        \begin{align*}
            t, \UPredSeq(\HComput) \lmodels \HPredTerm \Leftrightarrow \HPredTerm \in \UPredSeq(\HComput)_t \Leftrightarrow t, \emptyset, \HComput \models \HPredTerm
        \end{align*}
        \item Case $\varphi = \Upd{c_\pi}{\HFuncTerm{}}$
        \begin{align*}
            t, \UPredSeq(\HComput) \lmodels \Upd{c_\pi}{\HFuncTerm{}} \Leftrightarrow \HPredTerm \in \UPredSeq(\HComput)_t \Leftrightarrow t, \emptyset, \HComput \models \Upd{c_\pi}{\HFuncTerm{}}
        \end{align*}
        \item Case $\varphi = \neg \psi$
        \begin{align*}
            t, \UPredSeq(\HComput) \lmodels \neg \psi \Leftrightarrow \neg(t, \UPredSeq(\HComput) \lmodels \psi) \Leftrightarrow \neg(t, \emptyset, \HComput \models \psi) \Leftrightarrow t, \emptyset, \HComput \models \neg \psi
        \end{align*}
        \item Case $\varphi = \psi \wedge \psi'$
        \begin{align*}
            & &&t, \UPredSeq(\HComput) \lmodels \psi \wedge \psi' \\
            &\Leftrightarrow &&t, \UPredSeq(\HComput) \lmodels \psi \wedge t, \UPredSeq(\HComput) \lmodels \psi' \\
            &\Leftrightarrow &&t, \emptyset, \HComput \models \psi \wedge t, \emptyset, \HComput \models \psi' \\
            &\Leftrightarrow &&t, \emptyset, \HComput \models \psi \wedge \psi'
        \end{align*}
        \item Case $\varphi = \LTLnext \psi$
        \begin{align*}
            t, \UPredSeq(\HComput) \lmodels \LTLnext \psi \Leftrightarrow t+1, \UPredSeq(\HComput) \lmodels \psi \Leftrightarrow t+1, \emptyset, \HComput \models \psi \Leftrightarrow t, \emptyset, \HComput \models \LTLnext \psi
        \end{align*}
        \item Case $\varphi = \psi \LTLuntil \psi'$
        \begin{align*}
            & &&t, \UPredSeq(\HComput) \lmodels \psi \LTLuntil \psi'\\
            & \Leftrightarrow &&\exists t' \geq t.~ t', \UPredSeq(\HComput) \lmodels \psi' \wedge \forall t \leq t'' < t'.~t'', \UPredSeq(\HComput) \lmodels \psi \\
            &\Leftrightarrow &&\exists t' \geq t.~ t', Z, \HComput \models \psi' \wedge \forall t \leq t'' < t'.~t'', Z, \HComput \models \psi \\
            &\Leftrightarrow &&t, \emptyset, \HComput \models \psi \LTLuntil \psi'
        \end{align*}       
    \end{itemize} 
    \vspace{-0.5cm}
\end{proof}

\subsection{Proof of Theorem \ref{thm:k_feasible}} \label{sec:k_feasible_proof}
\begin{proof}
    $\Rightarrow$
    Let $q_0, q_1, q_2 \dots \in Q^\omega$ be a run of $P$ on the $k$-feasible trace $\PTrace$. Then, for every $j \in \mathbb{N}$, $$e_j = ((q_j,\PTrace_{j},q_{j+1} \dots ,\PTrace_{j+k-2},q_{j+k-1}), \PTrace_{j+k-1}, (q_{j+1},\PTrace_{j+1}, \dots, q_{j+k-1},\PTrace_{j+k-1},q_{j+k}))$$ is a transition of $P_k$. Moreover, for every $k' < k$,
    $$ e_{k'} = ((q_0,s_0,q_1\dots ,\PTrace_{k'-1},q_{k'}), \PTrace_{k'}, (q_0, \PTrace_0, \dots ,q_{k'},\PTrace_{k'},q_{k'+1})) $$
    is also a transition of $P_k$. Thus, $q_0, q_1, \dots$ is accepted by $P_k$.

    $\Leftarrow$ 
    Let $\PTrace$ be a trace of $P$ accepted by $P_k$. Then, there exist states of $P$ $q_0, q_1 \dots$ such that for every $j \in \mathbb{N}$, $e_j$ from above is a transition of $P_k$. Thus, by the definition of $P_k$ for every $j$, $\PTrace_j \dots \PTrace_{j+k-1}$ is feasible. Thus, $\PTrace$ is $k$-feasible.
    \qed 
\end{proof}

\subsection{Proof of Theorem \ref{thm:BuechiProdModelChecking}} \label{sec:buechi_corr}

The main idea of the correctness proof is a construction that, given a computation $\Comput$ that matches a program trace $\PTrace$, constructs a computation matching the combined trace $\combine{\PTrace}{\UPredSeq(\Comput)}$ and vice versa ($\UPredSeq$ was defined in Definition \ref{def:TSL_LTL}). This gives us the necessary feasibility proofs. To do so, we define two operations, $\widetilde{(-)}$ and $(-)_{|\PTrace}$ that `nearly' invert each other: we have that $(\widetilde{\Comput})_{|\PTrace} = \Comput$ and if $\matches{\Comput}{\combine{\PTrace}{X}}$ for some $X$, we also have that $\widetilde{\Comput_{|\PTrace}} = \Comput$. In Lemma \ref{lem:corr1} we show that if $\matches{\Comput}{\combine{\PTrace}{X}}$ for some $X$, then $\matches{\Comput_{|\PTrace}}{\PTrace}$. In Lemma \ref{lem:corr2} we show that then, we also have that $X=\UPredSeq(\Comput_{|\PTrace})$. Lemma \ref{lem:corr3} states the other direction: if $\matches{\Comput}{\PTrace}$, then also $\matches{\widetilde{\Comput}}{\combine{\PTrace}{\UPredSeq(\Comput)}}$. Those three lemmata give us the feasibility proofs needed for the algorithm's correctness. Lemma \ref{lem:TSL_LTL} then gives the equivalence between the violation of the TSL-formula by $\Comput$ and the sequence $\UPredSeq(\Comput)$ being accepted by $A_{\neg \varphi}$, needed for reasoning about the existence of a trace $\combine{\PTrace}{\UPredSeq(\Comput)}$ in the Büchi program product.

We start with definining the operation $\widetilde{(-)}$. Let $\PTrace \in \Stmt^\omega$ and $\matches{\Comput}{\PTrace}$. We need to extend this computation to one that matches $\combine{\PTrace}{\UPredSeq(\Comput)}$. For every time point $t$, we need to introduce computation steps that match $\combine{\PTrace_t}{\UPredSeq(\Comput)_t} =$ \\
$ \mathit{save\_values}_{{\UPredSeq(\Comput)_t}};~ \PTrace_t;~\mathit{new\_inputs};~\mathit{check\_preds}_{{\UPredSeq(\Comput)_t}};~ \mathit{check\_updates}_{{\UPredSeq(\Comput)_t}}$. While executing $\mathit{save\_values}_{{\UPredSeq(\Comput)_t}}$, the values of the temporary variables are changed as required by the statements $tmp_j := \FuncTerm{j}$. When the actual statement $\PTrace_t$ is executed, the computation changes to $\Comput_t$, but still with the `old' input values and extended with values for the temporal variables. Next, when executing $\mathit{new\_inputs}$, we stepwise change the input values to those in $\Comput_t$.  Then, the assertions are executed and the computation cannot change anymore. 

In the following, we also need the notion of extending an assignment: we define $a[c \mapsto v](c) = v$ and $a[c \mapsto v](c') = a(c')$ for $c \neq c'$.

Let $\USet \subseteq \UpdTerms$ be in the following the set of update terms, and $\PredSet \subseteq \PredTerms$ the set of predicate terms appearing in the formula $\varphi$. 

\begin{definition} \label{def:adaptedComput}
    Let $\Inputs = \{i_1, \dots i_n\}$ be the set of inputs and \\ $\USet = \{\Upd{c_1}{\FuncTerm{1}}, \dots ,\Upd{c_m}{\FuncTerm{m}}\}$. Given a computation $\Comput$, we define the \textbf{adapted computation} $\widetilde{\Comput}$ as follows.
    \begin{align*}
        \Assign^{\mathit{tmp}_1}_t &:= \Comput_{t-1} [\mathit{tmp}_1 \mapsto \Eval(\FuncTerm{1}, \Comput_{t-1})] \\
        \Assign^{\mathit{tmp}_j}_t &:= \Assign^{\mathit{tmp}_{j-1}} [\mathit{tmp}_j \mapsto \Eval(\FuncTerm{j}, \Comput_{t-1})] &&\text{for } 1 < j \leq m\\
        \Assign_t &:= \Comput_t[\mathit{tmp}_1 \mapsto \Eval(\FuncTerm{1}, \Comput_{t-1}), \dots , \mathit{tmp}_m \mapsto \Eval(\FuncTerm{m}, \Comput_{t-1}),\\
        &~~~~~~~~i_1 \mapsto \Comput_{t-1}(i_1),~ \dots,~i_n \mapsto \Comput_{t-1}(i_n) ] \\
        \Assign^{i_1}_t &:= \Assign_t [i_1 \mapsto \Comput_t(i_1)] \\
        \Assign^{i_j}_t &:= \Assign^{i_{j-1}} [i_j \mapsto \Comput_t(i_j)] &&\text{for } 1 < j \leq n\\
        \widetilde{\Comput_t} &:= a^{\mathit{tmp}_1}_t \dots a^{\mathit{tmp}_m}_t~a_t~a^{i_1}_t~\dots~a^{i_n}_t~a^{i_n}_t~a^{i_n}_t \\
        \widetilde{\Comput} &:= \widetilde{\Comput_0}~ \widetilde{\Comput_1}\dots 
    \end{align*}
\end{definition}

Note that this is the only possibility to adapt a computation $\matches{\Comput}{\sigma}$ such that the result \textit{could} match $\combine{\sigma}{X}$ for any $X$.

Also note that $a_t^{i_n} = \Comput_t [tmp_1 \mapsto \Eval(\FuncTerm{1}, \Comput_{t-1}), \dots , tmp_m \mapsto \Eval(\FuncTerm{m}, \Comput_{t-1})]$.

We can also define the left inverse of this operation: reducing a computation that matches $\combine{\PTrace}{X}$ to a computation that matches $\PTrace$.

\begin{definition}
    Let $\PTrace \in \Stmt^\omega, X \in \mathcal{P}(\PredTerms \cup \UpdTerms)^\omega$ and $\matches{\Comput}{\combine{\PTrace}{X}}$. We define the \textbf{reduced computation} $\Comput_{|\sigma}$ as follows.
    \begin{align*}
        \iota(j) &:= (|\Inputs| + |\USet| + 3) \cdot (j+1) -3 \\
        \Comput_{|\PTrace}(U) &:= (\Comput_{\iota(0)})_{| (\Inputs \cup \Cells)}~(\Comput_{\iota(1)})_{| (\Inputs \cup \Cells)} \dots
    \end{align*}
    where $a_{| (\Inputs \cup \Cells)}$ means restricting the domain of the assignment to the original inputs and cells, thus excluding the temporal variables $tmp_1, tmp_2 \dots$
\end{definition}

Note that if $\matches{\Comput}{\combine{\PTrace}{X}}$, we also have that $\Comput = \widetilde{\Comput_{|\PTrace}}$, as this is the \textit{only} computation that could potentially match $\combine{\PTrace}{X}$ and equals $\Comput_{|\PTrace}$ when restricted to~$\PTrace$. 

\begin{lemma} \label{lem:corr1}
    If $\matches{\Comput}{\combine{\PTrace}{X}}$, then $\matches{\Comput_{|\PTrace}}{\sigma}$.
\end{lemma}
\begin{proof}
    We have to show that $\forall t \in \mathbb{N}.~\matchest{\Comput_{|\PTrace}}{\PTrace}{t}$

    Recall that $\Comput = \widetilde{\Comput_{|\PTrace}}$ and
    \begin{align*}
        \widetilde{(\Comput_{|\PTrace})_t} = a^{tmp_1}_t \dots a^{tmp_m}_t~a_t~a^{i_1}_t~\dots~a^{i_n}_t~a^{i_n}_t~a^{i_n}_t
    \end{align*}

    \begin{itemize}
        \item Case $\PTrace_t = \mathit{assert}(\PredTerm)$ \\
        We know that $\matchest{\Comput}{\combine{\PTrace}{X}}{\iota(t)-|\Inputs|}$. The corresponding statement is $\sigma_j$, thus
        \begin{align*}
            \Eval(\PredTerm, \Comput_{\iota(t)-|\Inputs|-1}) = true ~~~\wedge \forall c \in \Cells^*.~\Comput_{\iota(t)-|\Inputs|}(c) = \Comput_{\iota(t)-|\Inputs|-1}(c)
        \end{align*} 
        Moreover, $(\Comput_{\iota(t)-|\Inputs|-1}) = a_t^{tmp_m}$. This equals $(\Comput_{|\PTrace})_{t-1}$ extended with values for the temporary variables. As $\PredTerm$ does not contain the temporal variables, this means that $\Eval(\PredTerm, ((\Comput_{|\PTrace})_{t-1})_{|(\Inputs \cup \Cells)})$ is also true. It remains to show that             
        \begin{align*}
            \forall c \in \Cells.~(\Comput_{\iota(t-1)})_{|(\Inputs \cup \Cells)} (c) = (\Comput_{\iota(t)})_{|(\Inputs \cup \Cells)} (c)
        \end{align*}
        This is true as the only cells changed in $\Comput_{\iota(t-1)} \dots \Comput_{\iota(t)-|\Inputs|-1}$ and in $\Comput_{\iota(t) - |\Inputs|},\dots \Comput_{\iota(t)}$ are cells from $\Cells^* \backslash \Cells$.
        \item The two remaining cases are analogous.
    \end{itemize}
    \vspace{-0.5cm}
\end{proof}

\begin{lemma} \label{lem:corr2}
    If $\matches{\Comput}{combine(\PTrace, X)}$, then $X = \UPredSeq(\Comput_{|\PTrace})$.
\end{lemma}
\begin{proof}
    We prove $\forall t.~ X_t = \UPredSeq(\Comput_{|\PTrace})_t$.
    We know that $\matchest{\Comput}{\combine{\PTrace}{X}}{\iota(t)+1}$. The corresponding statement is $\mathit{check\_preds}_{X_t}$. Set $h = \left(\bigwedge_{\PredTerm \in {X_t}} \PredTerm \wedge \bigwedge_{\PredTerm \in \PredSet \backslash {X_t}} \neg \PredTerm \right)$. This means that 
    \begin{align*}
        \Eval(h, \Comput_{\iota(t)+1}) = true ~~~\wedge \forall c \in \Cells^*.~\Comput_{\iota(t)+1}(c) = \Comput_{\iota(t)}(c)
    \end{align*}
    This implies that $true = \Eval((h, \Comput_{\iota(t)})_{|(\Inputs \cup \Cells)}) = \Eval(h, (\Comput_{|\PTrace})_t)$. Therefore, for all $\PredTerm \in \PredSet$
    \begin{align*}
        \PredTerm \in \UPredSeq(\Comput_{|\PTrace})_t \Leftrightarrow t, \Comput_{|\PTrace} \models \PredTerm \Leftrightarrow \Eval(\PredTerm, \Comput_{\iota(t)}) = true \Leftrightarrow \PredTerm \in {X_t}
    \end{align*} 

    For the update terms, we know that $\matchest{\Comput}{\combine{\PTrace}{X}}{\iota(t)+2}$. The corresponding statement is $\mathit{check\_updates}_{X_t}.$ Set $h=\left( \bigwedge_{\Upd{c_j}{\FuncTerm{j}} \in \USet}
    \begin{cases}
        c_j = tmp_j &\text{if } \Upd{c_j}{\FuncTerm{j}} \in X_t\\
        c_j \neq tmp_j &\text{else} 
    \end{cases}        
        \right)$ As before, we know that $\Eval(h, (\Comput_{|\PTrace})_t) = true$. Moreover, we know that for each $j$, $(\Comput_{|\PTrace})_t (tmp_j) = \Eval(\FuncTerm{j}, (\Comput_{|\PTrace})_{t-1})$ by definition \ref{def:adaptedComput} Therefore, for every $\Upd{c_j}{\FuncTerm{j}} \in \USet,$
        \begin{align*}
            \Upd{c_j}{\FuncTerm{j}} \in \UPredSeq(\Comput_{|\PTrace})_t &\Leftrightarrow t, \Comput_{|\PTrace} \models \Upd{c_j}{\FuncTerm{j}} \\ 
            &\Leftrightarrow \Eval({\FuncTerm{j}}, \Comput_{\iota(t-1)}) = \Eval(c_j, \Comput_{\iota(t)}) \\
            &\Leftrightarrow \Eval(c_j = tmp_j, \Comput_{\iota(t)}) = true \\
            &\Leftrightarrow \Upd{c_j}{\FuncTerm{j}} \in X_t
        \end{align*}
\end{proof}

\begin{lemma} \label{lem:corr3}
    If $\matches{\Comput}{\PTrace}$, then $\matches{\widetilde{\Comput}}{\combine{\PTrace}{\UPredSeq(\Comput)}}$
\end{lemma}
\begin{proof}
    We have to show that for all $t$, $\matchest{\widetilde{\Comput}}{\combine{\PTrace}{\UPredSeq(\Comput)}}{t}$. This is clear for all time steps except for those of kind $\mathit{check\_preds}$ or $\mathit{check\_updates}$ by the definition of $\widetilde{\Comput}$.

    First consider $\mathit{check\_preds}$. We need to show that $\forall t$, $\matchest{\widetilde{\Comput}}{\combine{\PTrace}{\UPredSeq(\Comput)}}{\iota(t)+1}$. This boils down to
    \begin{align*}
        \Eval\left( \left(\bigwedge_{\PredTerm \in {\UPredSeq(\Comput)_t}} \PredTerm \wedge \bigwedge_{\PredTerm \in \PredSet \backslash {\UPredSeq(\Comput)_t}} \neg \PredTerm \right), \widetilde{\Comput}_{\iota(t)} \right) = true
    \end{align*}
    As the temporary variables $\mathit{tmp}_1, \mathit{tmp}_2 \dots $ are not used in any $\PredTerm \in \PredSet$, this is by definition \ref{def:adaptedComput} equivalent to
    \begin{align*}
        \forall \PredTerm \in \PredSet.~\PredTerm \in \UPredSeq(\Comput)_t \Leftrightarrow \Eval(\PredTerm, \Comput_t) = true
    \end{align*}
    This is true by the definition of $\UPredSeq(\Comput)_t$.

    Now consider $\mathit{check\_updates}$. We need to show that $\forall t$, $\matchest{\widetilde{\Comput}}{\combine{\PTrace}{\UPredSeq(\Comput)}}{\iota(t)+2}$. This boils down to
    \begin{align*}
        \Eval\left(\left( \bigwedge_{\Upd{c_j}{\FuncTerm{j}} \in \USet}
        \begin{cases}
            c_j = \mathit{tmp}_j &\text{if } \Upd{c_j}{\FuncTerm{j}} \in \UPredSeq(\Comput)_t\\
            c_j \neq \mathit{tmp}_j &\text{else} 
        \end{cases}        
            \right), \widetilde{\Comput}_{\iota(t)} \right) = true
    \end{align*}
    Which is equivalent to 
    \begin{align*}
        \forall \Upd{c_j}{\FuncTerm{j}} \in \USet.~ \Eval(c_j = \mathit{tmp}_j, \widetilde{\Comput}_{\iota(t)}) = true \Leftrightarrow \Upd{c_j}{\FuncTerm{j}} \in \UPredSeq(\Comput)_t
    \end{align*}
    We know that $\widetilde{\Comput}_{\iota(t)}(\mathit{tmp}_j) = \Eval(\FuncTerm{j}, \Comput_{t-1})$. Thus this is equivalent to 
    \begin{align*}
        \forall \Upd{c_j}{\FuncTerm{j}} \in \USet.~ \Eval(\FuncTerm{j}, \Comput_{t-1}) = \Comput_{t}(c) \Leftrightarrow \Upd{c_j}{\FuncTerm{j}} \in \UPredSeq(\Comput)_t
    \end{align*}
    which is again true by the definition of $\UPredSeq(\Comput)_t$.
\end{proof}

Now, we have all the lemmas needed to prove Theorem \ref{thm:BuechiProdModelChecking}
\begin{proof} (Theorem \ref{thm:BuechiProdModelChecking}) \\
    $\Rightarrow$
    Assume that $P \BuechiProd A_{\neg \varphi}$ has a feasible trace. Then, this is a trace $\combine{\PTrace}{X}$ for some $\PTrace \in \mathcal{L}(P)$ and $X \in \mathcal{L}(A_{\neg \varphi})$. Moreover, $\matches{\Comput}{\combine{\PTrace}{X}}$ for some $\Comput \in \Assigns^\omega$. By Lemma \ref{lem:corr2}, we know that $X=\UPredSeq(\Comput)$ and by Lemma \ref{lem:corr1} we know that $\matches{\Comput}{\PTrace}$. By the correctness of $A_\varphi$, we know that $\UPredSeq(\Comput) \lmodels \neg \varphi$, which by Lemma \ref{lem:TSL_LTL} means that $\Comput \models \neg \varphi$. Thus $\Comput$ is a counterexample that proves that $P$ does not satisfy $\varphi$.

    $\Leftarrow$
    Assume that $P$ does not satisfy $\varphi$. Then, there is a trace $\PTrace \in \mathcal{L}(P)$ and a computation $\Comput$ such that $\matches{\Comput}{\PTrace}$ and $\Comput \models \neg \varphi$. This means by Lemma \ref{lem:TSL_LTL} that $\UPredSeq(\Comput) \lmodels \neg \varphi$, so $\UPredSeq(\Comput)$ is accepted by $A_{\neg \varphi}$. Then, $\combine{\PTrace}{\UPredSeq(\Comput)}$ is a trace of $P \BuechiProd A_{\neg \varphi}$. By Lemma \ref{lem:corr3}, $\matches{\widetilde{\Comput}}{\combine{\PTrace}{\UPredSeq(\Comput)}}$, so this is also a feasible trace.
\end{proof}

\subsection{Proof of Theorems \ref{thm:BuechiProdModelCheckingAFH} and \ref{thm:BuechiProdModelCheckingAFH2}} \label{sec:BuechiProd_corr2}
 As the two theorems are dual, it suffices to give the proof for Theorem \ref{thm:BuechiProdModelCheckingAFH}.
 
The proof is analogous to the proof of Theorem \ref{thm:BuechiProdModelChecking}, but we have to deal with multiple traces and thus even more indices now. We give it here for completeness.

Given $n$ program traces $\PTrace_{\pi_1}, \dots \PTrace_{\pi_n}$, we define $\PTrace_j = ({(\PTrace_{\pi_1})_{\pi_1}}_j; {(\PTrace_{\pi_2})_{\pi_2}}_j; \dots {(\PTrace_{\pi_n})_{\pi_n}}_j)$ and $\PTrace = \PTrace_1 \PTrace_2 \dots$. Let $\matches{\Comput_{\pi_1}}{\PTrace_{\pi_1}} \wedge \dots \wedge \matches{\Comput_{\pi_n}}{\PTrace_{\pi_n}}$. Let $\HComput = \ExtComput{\emptyset}{\pi_1}{\Comput_{\pi_1}} \dots \ExtComput{}{\pi_n}{\Comput_{\pi_n}}$. Those computations are extendable to a computation that matches $\combine{\PTrace}{\UPredSeq(\HComput)}$ For every time point $t$, we need to introduce the computation steps that match $\combine{\PTrace_t}{X_t} = \mathit{save\_values};~\PTrace_t;~\mathit{new\_inputs};~\mathit{check\_preds}_{X_t};~ \mathit{check\_updates}_{X_t}$. While executing $\mathit{save\_values}$, the values of the relevant temporary variables are changed as required by the statements $tmp_j := \HFuncTerm{j}$. After the actual statements $\PTrace_t$ are executed, the computation changes to $\HComput_t$, but still with the `old' inputs and extended with values for the temporal variables. Next, when executing $\mathit{new\_inputs}$, we stepwise change the input values to those in $\HComput_t$. Then, the assertions are executed and the computation cannot change anymore. 

In the following, we also need the notion of extending a hyper-assignment: we define $\HAssign[c \mapsto v](c) = v$ and $\HAssign[c \mapsto v](c') = \HAssign(c')$ for $c \neq c'$.

Let $\USet \subseteq \HUpdTerms$ be in the following the set of update terms and $\PredSet \subseteq \HPredTerms$ the predicate terms appearing in the formula $\varphi$.

\begin{definition} \label{def:hadaptedComput}
    Let $\Inputs \times \TraceVs = \{i_1, \dots i_k\}$ be the set of inputs and $\USet = \{\Upd{c_1}{\HFuncTerm{1}}, \dots ,\Upd{c_m}{\HFuncTerm{m}}\}$. Given computations $\Comput_{\pi_1} \dots \Comput_{\pi_n}$, let $\HComput = \ExtComput{\emptyset}{\pi_1}{\PTrace_{\pi_1}} \dots \ExtComput{}{\pi_n}{\PTrace_{\pi_n}}$. We define the \textbf{adapted computation} $\widetilde{(\Comput_{\pi_1}, \dots ,\Comput_{\pi_n})}$.
    \begin{align*}
        \HAssign^{\mathit{tmp}_1}_t &:= \HComput_{t-1} [\mathit{tmp}_1 \mapsto \Eval(\FuncTerm{1},\HComput_{t-1})] \\
        \HAssign^{\mathit{tmp}_j}_t &:= \HAssign^{\mathit{tmp}_{j-1}} [\mathit{tmp}_j \mapsto \Eval(\FuncTerm{j},\HComput_{t-1})] ~~~~~~~~~~~~~~~~~~~\text{for } 1 < j \leq m\\
        \HAssign^{\pi_j}_t &= (\HComput_{t-1}\ExtComput{}{\pi_1}{\Comput_{\pi_1}} \dots \ExtComput{}{\pi_j}{\Comput_{\pi_j}})_t
         [i_1 \mapsto \Comput_{t-1}(i_1),~\dots~,i_k \mapsto \Comput_{t-1}(i_k), \\
         &~~~~\mathit{tmp}_1 \mapsto \Eval(\FuncTerm{1},\HComput_{t-1}), \dots , \mathit{tmp}_m \mapsto \Eval(\FuncTerm{m},\HComput_{t-1})] \\
         \HAssign^{i_1}_t &:= \HAssign_t^{\pi_n} [i_1 \mapsto \HComput_t(i_1)] \\
        \HAssign^{i_j}_t &:= \HAssign^{i_{j-1}} [i_j \mapsto \HComput_t(i_j)] ~~~~~~~~~~~~~~~~~~~~~~~~~~~~\text{for } 1 < j \leq k\\
         \widetilde{(\Comput_{\pi_1}, \dots ,\Comput_{\pi_n})^t} &:= \HAssign^{\mathit{tmp}_1}_t \dots \HAssign^{\mathit{tmp}_n}_t ~\HAssign^{\pi_1}_t \dots \HAssign^{\pi_n}_t~\HAssign^{i_1}_t~\dots~\HAssign^{i_k}_t~\HAssign^{i_k}_t~\HAssign^{i_k}_t \\
         \widetilde{(\Comput_{\pi_1}, \dots ,\Comput_{\pi_n})} &:= \widetilde{(\Comput_{\pi_1}, \dots ,\Comput_{\pi_n})^0}~ \widetilde{(\Comput_{\pi_1}, \dots ,\Comput_{\pi_n})^1}\dots 
    \end{align*}
\end{definition}

Note that this is the only possibility to extend the computations $\matches{\Comput_{\pi_1}}{\PTrace_{\pi_1}}, \dots, \matches{\Comput_{\pi_n}}{\PTrace_{\pi_n}}$ to a computation that potentially matches $\combine{\PTrace}{X}$ for any $X$.

We can also define the left inverse of this operation: reducing a computation that matches $\combine{\PTrace}{X}$ to computations that match $\PTrace_{\pi_1}, \dots , \PTrace_{\pi_n}$ as follows.

\begin{definition}
    Let $1 \leq j \leq n, \PTrace \in \Stmt^\omega, X \in \mathcal{P}(\HPredTerms \cup \HUpdTerms)^\omega$ and $\matches{\HComput}{\combine{\PTrace}{X}}$.
    We define the index of the computation step of $(\PTrace_{\pi_{j}})_t$ in $\combine{\PTrace}{X}$
    \begin{align*}
        \iota(t) &:= (|\Inputs \times \TraceVs| + |\USet| + n + 2) \cdot (t+1) - 3
    \end{align*}
    We define the \textbf{reduced computation} $\Comput_{|\pi_j}$.
    \begin{align*}
        \Comput_{|\pi_j} &:= (\Comput_{\iota(0)})_{| (\Inputs \cup \Cells)_{\pi_j}} ~(\Comput_{\iota(1)})_{| (\Inputs \cup \Cells)_{\pi_j}} \dots
    \end{align*}
    where $\HAssign_{| (\Inputs \cup \Cells)_{\pi_j}}$ means restricting the domain of the assignment to the cells and inputs labeled with $\pi_j$, thus excluding the temporal variables $tmp_1, tmp_2 \dots$ and the variables from other traces. Moreover, the cells and inputs are again renamed from $c_{\pi_j}$ to $c$ or $i_{\pi_j}$ to $i$.
\end{definition}

Note that if $\matches{\HComput}{\combine{\PTrace}{X}}$, we also have that $\HComput$ is the adapted computation of $(\HComput_{|\pi_1}, \dots ,\HComput_{|\pi_n})$ as this is the \textit{only} computation that could match $\combine{\PTrace}{X}$ and equals $\Comput_{|\pi_j}$ when restricted to $\pi_j$. 

\begin{lemma} \label{lem:hcorr1}
    If $\matches{\HComput}{\combine{\PTrace}{X}}$ and $\PTrace = ((\PTrace_{\pi_1})_{\pi_1}, \dots ,(\PTrace_{\pi_n})_{\pi_n})$, then $\matches{\HComput_{|\pi_j}}{\PTrace_{\pi_j}}$ for every $1 \leq j \leq n$.
\end{lemma}
\begin{proof}
    We show that $\forall t \in \mathbb{N}. ~\matchest{\HComput_{|\pi_j}}{\PTrace_{\pi_j}}{t}$.

    Recall that $\HComput$ is the adapted computation of $(\HComput_{|\pi_1}, \dots ,\HComput_{|\pi_n})$ .

    \begin{itemize}
        \item Case $\PTrace_{t} = \mathit{assert}(\PredTerm)$ \\
        We know that $\matchest{\HComput}{\combine{\PTrace}{X}}{\iota(t) - |\Inputs \times \TraceVs| - (n-j)}$, and  thus
        \begin{align*}
            &\Eval(\HPredTerm, \HComput_{\iota(t)-|\Inputs \times \TraceVs|-(n-j)-1}) = true~~ \wedge \\ &\forall c \in \Cells^*.~\HComput_{\iota(t)-|\Inputs \times \TraceVs| - (n-j)}(c) = \HComput_{\iota(t)-|\Inputs \times \TraceVs|(n-j)-1}(c)
        \end{align*} 
        Moreover, $(\HComput_{\iota(t)-|\Inputs \times \TraceVs|-(n-j)-1})$ equals $\HAssign^{tmp_m}_t$ if $j=0$ and else $\HAssign^{\pi_{j-1}}_t$, which both equals $(\HComput_{|\pi_j})_{t-1}$ when restricted to the inputs and variables from $\pi_j$. $\PredTerm$ does not contain variables from other traces or temporary variables, thus $\Eval(\PredTerm, (\HComput_{|\pi_j})_{t-1})$ is also true. It remains to show that             
        \begin{align*}
            \forall c \in \Cells.~((\HComput)_{\iota(t-1)})_{|(\Inputs \cup \Cells)_{\pi_j}} (c) = ((\HComput)_{\iota(t)})_{|(\Inputs \cup \Cells)_{\pi_j}} (c)
        \end{align*}
        This is also true as the only cells changed in $\HComput_{\iota(t-1)} \dots \HComput_{\iota(t)-|\Inputs \times \TraceVs|-(n-j)-1}$ and in $\HComput_{\iota(t)-(n-j)-|\Inputs \times \TraceVs|}, \dots \HComput_{\iota(t)}$ are cells from $\Cells^* \backslash \Cells$ or cells from other traces.
        \item The two remaining cases are analogous.
    \end{itemize}
\end{proof}

\newcommand{\longseq}{\UPredSeq(\emptyset \ExtComput{}{\pi_1}{\HComput_{|\pi_1}} \dots \ExtComput{}{\pi_n}{\HComput_{|\pi_n}} )}

\begin{lemma} \label{lem:hcorr2}
    If $\matches{\HComput}{combine(\PTrace, X)}$, then $X = \longseq$ 
\end{lemma}
\begin{proof}
    Set $\HComput' = \emptyset \ExtComput{}{\pi_1}{\HComput_{|\pi_1}} \dots \ExtComput{}{\pi_n}{\HComput_{|\pi_n}}$. We prove $\forall t.~ X_t = \UPredSeq(\HComput')_t$.
    We know that $\matchest{\HComput}{\combine{\PTrace}{X}}{\iota(t)+1}$. The corresponding statement is $check\_preds_{X_t}$. Set $h = \left(\bigwedge_{\PredTerm \in {X_t}} \HPredTerm \wedge \bigwedge_{\HPredTerm \in \PredSet \backslash {X_t}} \neg \PredTerm \right)$. This means that 
    \begin{align*}
        \Eval(h, \HComput_{\iota(t)+1}) = true ~~~\wedge \forall c \in \Cells^*.~\Comput_{\iota(t)+1}(c) = \HComput_{\iota(t)}(c)
    \end{align*}

    Recall that $\HComput_{\iota(t)+1}$ is by Definition \ref{def:hadaptedComput} equal to
    \begin{align*}
        &(\HComput'_{t-1}\ExtComput{}{\pi_1}{\HComput'_{\pi_1}} \dots \ExtComput{}{\pi_j}{\HComput'_{\pi_n}})_t~
         [\mathit{tmp}_1 \mapsto \Eval(\HFuncTerm{1}, \HComput'_{t-1}), \dots , \mathit{tmp}_m \mapsto \Eval(\HFuncTerm{m}, \HComput'_{t-1})] \\
         &=  \HComput'_t~
         [\mathit{tmp}_1 \mapsto \Eval(\HFuncTerm{1}, \HComput'_{t-1}), \dots , \mathit{tmp}_m \mapsto \Eval(\HFuncTerm{m}, \HComput'_{t-1})]
    \end{align*}
    $h$ does not contain the temporal variables, so this implies that $\Eval(h, \HComput'_t) = true$. Therefore, for all $\PredTerm \in P$
    \begin{align*}
        \PredTerm \in \UPredSeq(\HComput')_t \Leftrightarrow t, \HComput' \models \PredTerm \Leftrightarrow \Eval(\PredTerm, \Comput'_t) = true \Leftrightarrow \PredTerm \in {X_t}
    \end{align*} 
    The last equivalence holds by the definition of $h$.

    For the update terms, we know that $\matchest{\Comput}{\combine{\PTrace}{X}}{\iota(t)+2}$. The corresponding statement is $\mathit{check\_updates}_{X_t}$. Set $h=\left( \bigwedge_{\Upd{c_j}{\HFuncTerm{j}} \in \USet}
    \begin{cases}
        c_j = \mathit{tmp}_j &\text{if } \Upd{c_j}{\HFuncTerm{j}} \in X_t\\
        c_j \neq \mathit{tmp}_j &\text{else} 
    \end{cases}        
        \right)$ As before, we know that $\Eval(h, \HComput'_t) = true$. Moreover, we know that for each $j$, $\HComput_{\iota(t)+1} (\mathit{tmp}_j) = \Eval(\HFuncTerm{j},\HComput'_{t-1})$ again by Definition \ref{def:hadaptedComput} Therefore, for every $\Upd{c_j}{\HFuncTerm{j}} \in \USet,$
        \begin{align*}
            \Upd{c_j}{\HFuncTerm{j}} \in \UPredSeq(\HComput')_t &\Leftrightarrow t, \HComput' \models \Upd{c_j}{\HFuncTerm{j}} \\ 
            &\Leftrightarrow \Eval({\HFuncTerm{j}}, \HComput'_{t-1}) = \Eval(c_j, \HComput'_t) \\
            &\Leftrightarrow \Eval(c_j = tmp_j, \HComput'_t) = true \\
            &\Leftrightarrow \Upd{c_j}{\HFuncTerm{j}} \in X_t
        \end{align*}
        The last equivalence is again true by the definition of $h$.
\end{proof}

\begin{lemma} \label{lem:hcorr3}
    If $\matches{\Comput_{\pi_1}}{\PTrace_{\pi_1}} \wedge \dots \wedge \matches{\Comput_{\pi_n}}{\PTrace_{\pi_n}}$, then $\matches{\widetilde{(\Comput_{\pi_1}, \dots ,\Comput_{\pi_n})}}{\combine{\PTrace}{\UPredSeq(\HComput')}$, where $\HComput' = \ExtComput{\emptyset}{\pi_1}{\Comput_{\pi_1}} \dots \ExtComput{}{\pi_n}{\Comput_{\pi_n}}}$
\end{lemma}
\begin{proof}
    Set $\HComput = \widetilde{(\Comput_{\pi_1}, \dots ,\Comput_{\pi_n})}$. We have to show that for all $t$, $\matchest{\HComput}{\combine{\PTrace}{\UPredSeq(\HComput')}}{t}$. This is clear for all time steps except for those of kind $\mathit{check\_preds}$ or $\mathit{check\_updates}$ by the definition of $\HComput$.

    First consider $\mathit{check\_preds}$. We need to show that $\forall t$, $\matchest{\HComput}{\combine{\PTrace}{\UPredSeq(\HComput')}}{\iota(t)+1}$. This boils down to
    \begin{align*}
        \Eval\left( \left(\bigwedge_{\HPredTerm \in {\UPredSeq(\HComput')_t}} \HPredTerm \wedge \bigwedge_{\HPredTerm \in \PredSet \backslash {\UPredSeq(\HComput')_t}} \neg \HPredTerm \right), \HComput_{\iota(t)+1} \right) = true
    \end{align*}

    Recall that $\HComput_{\iota(t)+1}$ is by Definition \ref{def:hadaptedComput} equal to
    \begin{align*}
        &(\HComput'_{t-1}\ExtComput{}{\pi_1}{\HComput'_{\pi_1}} \dots \ExtComput{}{\pi_j}{\HComput'_{\pi_n}})_t~
         [\mathit{tmp}_1 \mapsto \Eval(\HFuncTerm{1},\HComput'_{t-1}), \dots , \mathit{tmp}_m \mapsto \Eval(\HFuncTerm{m}, \HComput_{t-1})] \\
         &= \HComput'_t~ [\mathit{tmp}_1 \mapsto \Eval(\HFuncTerm{1},\HComput'_{t-1}), \dots , \mathit{tmp}_m \mapsto \Eval(\HFuncTerm{m}, \HComput_{t-1})] 
    \end{align*}

    Thus, as the temporary variables are not used in $\HPredTerm$, this is equivalent to
    \begin{align*}
        \forall \HPredTerm \in \PredSet.~\HPredTerm \in \UPredSeq(\Comput')_t \Leftrightarrow \Eval((\HPredTerm), \Comput'_t ) = true
    \end{align*}
    This is true by the definition of $\UPredSeq(\Comput)_t$.

    Now consider $\mathit{check\_updates}$. We need to show that $\forall t$, $\matchest{\HComput}{\combine{\PTrace}{\UPredSeq(\HComput')}}{\iota(t)+2}$. This boils down to
    \begin{align*}
        \Eval\left(\left( \bigwedge_{\Upd{c_j}{\HFuncTerm{j}} \in \USet}
        \begin{cases}
            c_j = \mathit{tmp}_j &\text{if } \Upd{c_j}{\HFuncTerm{j}} \in \UPredSeq(\HComput')_t\\
            c_j \neq \mathit{tmp}_j &\text{else} 
        \end{cases}        
            \right), \HComput_{\iota(t)+2} \right) = true
    \end{align*}
    Which is again equivalent to 
    \begin{align*}
        \forall \Upd{c_j}{\FuncTerm{j}} \in \USet.~ \Eval(c_j = \mathit{tmp}_j, \HComput_{\iota(t)+2}) = true \Leftrightarrow \Upd{c_j}{\FuncTerm{j}} \in \UPredSeq(\HComput')_t
    \end{align*}
    We know that $\HComput_{\iota(t)+2}(\mathit{tmp}_j) = \HComput'_{t-1}(\FuncTerm{j})$. Thus this is equivalent to 
    \begin{align*}
        \forall \Upd{c_j}{\HFuncTerm{j}} \in \USet.~ \Eval(\HFuncTerm{j}, \HComput'_{t-1}) = \HComput'_{t}(c) \Leftrightarrow \Upd{c_j}{\HFuncTerm{j}} \in \UPredSeq(\HComput')_t
    \end{align*}
    Which is again true by the definition of $\UPredSeq(\HComput')_t$.
\end{proof}

Now, we have all the lemmas needed to prove Theorem \ref{thm:BuechiProdModelCheckingAFH}
\begin{proof} (Theorem \ref{thm:BuechiProdModelCheckingAFH}) \\
    $\Rightarrow$
    Assume that $P^n \BuechiProd A_{\neg \varphi}$ has a feasible trace. Then, this is a trace $\combine{\PTrace}{X}$ for some $\PTrace \in \mathcal{L}(P^n)$ and $X \in \mathcal{L}(A_{\neg \varphi})$. We know that $\PTrace_t = ({(\PTrace_{\pi_1})_{\pi_1}}_t; \dots ;{(\PTrace_{\pi_n})_{\pi_n}}_t)$. Moreover, $\matches{\HComput}{\combine{\PTrace}{X}}$ for some $\HComput \in \HAssigns^\omega$. By Lemma \ref{lem:hcorr1}, we know that $\matches{\HComput_{|\pi_1}}{\PTrace_{\pi_1}} \wedge \dots \wedge \matches{\HComput_{|\pi_n}}{\PTrace_{\pi_n}}$. Set $\HComput' = \ExtComput{\emptyset}{\pi_1}{\HComput_{|\pi_1}} \dots \ExtComput{}{\pi_n}{\HComput_{|\pi_n}}$. By Lemma \ref{lem:hcorr2}, we know that $X=\UPredSeq(\HComput')$. By the correctness of $A_\psi$ this means that $\UPredSeq(\HComput') \lmodels \neg \varphi$ which by Lemma \ref{lem:TSL_LTL} means that $\HComput' \models \neg \varphi$. Thus, $\HComput_{|\pi_1} \dots \HComput_{|\pi_n}$ are feasible counterexample traces proving that $\forall \pi_1.~ \dots \forall \pi_n.~ \psi$ does not hold.

    $\Leftarrow$
    Assume that $P$ does not satisfy $\varphi$. Then, there are trace $\PTrace_{\pi_1}, \dots \PTrace_{\pi_n} \in \mathcal{L}(P)$ and computations $\Comput_{\pi_1}, \dots \Comput_{\pi_n}$ such that $\matches{\Comput_{\pi_1}}{\PTrace_{\pi_1}} \wedge \dots \wedge \matches{\Comput_{\pi_n}}{\PTrace_{\pi_n}}$ and \\ $\HComput' = \ExtComput{\emptyset}{\pi_1}{\Comput_{\pi_1}} \dots \ExtComput{}{\pi_n}{\Comput_{\pi_n}} \models \neg \varphi$. This means by Lemma \ref{lem:TSL_LTL} that $\UPredSeq(\HComput') \lmodels \neg \varphi$, so $\UPredSeq(\HComput')$ is accepted by $A_{\neg \varphi}$. Set $\PTrace_t = ({(\PTrace_{\pi_1})_{\pi_1}}_t; \dots ;~{(\PTrace_{\pi_n})_{\pi_n}}_t)$ and $\PTrace = \PTrace_0 \PTrace_1 \dots$. Then, $\combine{\PTrace}{\UPredSeq(\Comput)}$ is a trace of $P \BuechiProd A_{\neg \varphi}$. By Lemma \ref{lem:hcorr3}, $\matches{\widetilde{(\Comput_{\pi_1}, \dots, \Comput_{\pi_n})}}{\combine{\PTrace}{\UPredSeq(\Comput')}}$, so this is also a feasible trace.

\end{proof}

\subsection{Proof of Theorem \ref{thm:corr_AE}} \label{sec:corr_AE}

We now prove Theorem \ref{thm:corr_AE}. To do that, we need the following lemma: recall that for a program execution $\sigma_{\pi}$, $(\sigma_{\pi})_{\pi}$ means renaming every cell $c$ in $\sigma_\pi$ to $c_\pi$ and every input $i$ to $i_\pi$.

\begin{lemma} \label{lem:splittrace}
    Let $\PTrace \in \Stmt^\omega$ be feasible and $\PTrace_t = (((\PTrace_{\pi_1})_{\pi_1})_t, \dots , ((\PTrace_{\pi_n})_{\pi_n})_t )$ for some $\sigma_{\pi_1}, \dots \sigma_{\pi_n}$. Then $\PTrace_{\pi_1}, \dots, \PTrace_{\pi_n}$ are also all feasible.
\end{lemma}
\begin{proof}
    As $\PTrace$ is feasible, we know that $\matches{\HComput}{\PTrace}$ for some $\HComput$. For all $1 \leq j \leq n$, we define $\Comput_{\pi_j}$ by
    \begin{align*}
        (\Comput_{\pi_j})_t &= (\HComput_{t \cdot n + j - 1})_{|(\Inputs \cup \Cells)_{\pi_j}} \\
        \Comput_{\pi_j} &= (\Comput_{\pi_j})_0~(\Comput_{\pi_j})_1 \dots
    \end{align*}
    where $\HAssign_{| (\Inputs \cup \Cells)_{\pi_j}}$ as before means restricting the domain of the assignment to the cells and inputs labeled with $\pi_j$, thus excluding the variables from other traces. Moreover, the cells and inputs are again renamed from $c_{\pi_j}$ to $c$ or $i_{\pi_j}$ to $i$. $t \cdot m + j - 1$ is the index of $(\PTrace_{\pi_j})_t$ in $\PTrace$.
    
    We show that for all time points $t$, $\matchest{\Comput_{\pi_j}}{\PTrace_{\pi_j}}{t}$
    \begin{itemize}
        \item Case $(\Comput_{\pi_j})_t = \mathit{assert}(\PredTerm)$ \\
                We know that $\matchest{\HComput}{\PTrace}{t \cdot m + j - 1}$ and thus 
                \begin{align*}
                    \Eval(\HPredTerm, \HComput_{t \cdot m + j - 1}) = true \wedge \forall c \in \Cells \times \TraceVs.~\HComput_{t \cdot m + j - 1}(c) = \HComput_{t \cdot m + j - 2}.
                \end{align*}
                Moreover $\Eval(\PredTerm, (\Comput_{\pi_j})_t)$ is also true as $\PredTerm$ does not contain variables from other traces. It remains to show that 
                \begin{align*}
                    \forall c \in \Cells.~((\Comput_{\pi_j})_t)(c) = (\Comput_{\pi_j})_{t-1}(c)
                \end{align*}

                This is also true as the only cells changed in $\HComput_{(t-1) \cdot m + j - 1}, \dots \HComput_{t \cdot m + j - 2}$ are cells from other traces.
        \item The remaining two cases are analogous.
    \end{itemize}
\end{proof}

We now prove Theorem \ref{thm:corr_AE}
\begin{proof}
    Assume that $P^m \backslash (P^n \BuechiProd A_\psi)_{k, C(k')}^\forall$ has a feasible trace $\PTrace$ with $\matches{\HComput}{\PTrace}$. By Lemma \ref{lem:splittrace}, this means that $\matches{\Comput_{\pi_1}}{\PTrace_{\pi_1}} \wedge \dots \wedge \matches{\Comput_{\pi_m}}{\PTrace_{\pi_m}}$. It suffices to show that $\emptyset \ExtComput{}{\pi_1}{\Comput_1} \dots \ExtComput{}{\pi_m}{\Comput_m} \nmodels \exists \pi_{m+1}.~ \dots \exists \pi_n.~\psi$ as this implies that $\Comput_1, \dots \Comput_m$ are a counterexample proving that $P$ does not satifsfy $\varphi$.

    Proof by contradiction. Assume that $\emptyset \ExtComput{}{\pi_1}{\Comput_1} \dots \ExtComput{}{\pi_m}{\Comput_m} \models \exists \pi_{m+1}.~ \dots \exists \pi_n.~\psi$. Then, there are traces $\PTrace_{\pi_{m+1}}, \dots \PTrace_{\pi_n}$ and computations $\Comput_{\pi_{m+1}} \dots \Comput_{\pi_n}$ such that $\matches{\Comput_{\pi_{m+1}}}{\PTrace_{\pi_{m+1}}} \wedge \dots \wedge \matches{\Comput_{\pi_n}}{\PTrace_{\pi_n}}$ and $\HComput' = \emptyset \ExtComput{}{\pi_1}{\Comput_{\pi_1}} \dots \ExtComput{}{\pi_n}{\Comput_{\pi_n}} \models \psi$.\\ Set $\PTrace'_t = ({(\PTrace_{\pi_1})_{\pi_1}}_t; \dots ;{(\PTrace_{\pi_n})_{\pi_n}}_t)$ and $\PTrace' = \PTrace'_0~\PTrace'_1 \dots$. Now, by Lemma \ref{lem:TSL_LTL} and the correctness of $A_\psi$, we know that $\UPredSeq(\HComput')$ is accepted by $A_\psi$, thus $\combine{\PTrace'}{ \UPredSeq(\HComput')}$ is accepted by $P^n \BuechiProd A_{\psi}$. Moreover, by Lemma \ref{lem:hcorr2}, we know that $\combine{\PTrace'}{\UPredSeq(\HComput')}$ is also feasible, so it is also $k$-feasible and thus accepted by $(P^n \BuechiProd A_\psi)_k$. Moreover, does not end with an infeasible cycle and is thus also accepted by $(P^n \BuechiProd A_\psi)_{k, C(k')}$  But this means that $\PTrace$ is accepted by $(P^n \BuechiProd A_\psi)^\forall_{k, C(k')}$. and thus not by $P^m \backslash (P^n \BuechiProd A_\psi)^\forall_{k, C(k')}$. Contradiction.
\end{proof}